\documentclass[pra,floatfix,superscriptaddress,longbibliography,notitlepage]{revtex4-1}

\pdfoutput=1

\usepackage[utf8]{inputenc}
\usepackage{graphicx}
\usepackage{hyperref}
\usepackage{natbib}
\usepackage{subfigure}
\usepackage{array}
\usepackage{amsmath}
\usepackage{amsthm}
\usepackage{amsfonts}
\usepackage{amssymb}
\usepackage{cleveref}
\usepackage{calc}

\theoremstyle{plain} 
\newtheorem{theorem}{Theorem}

\newtheorem{claim}[theorem]{Claim}
\theoremstyle{definition}
\newtheorem{definition}{Definition}
\newtheorem{note}{Note}
\theoremstyle{remark}

\newcommand{\ket}[1]{\ensuremath{\left\vert #1 \right\rangle}}

\newcommand{\mbf}[1]{\ensuremath{\mathbf{#1}}}
\newcommand{\mddots}{\mbf{\ddots{}}}
\newcommand{\bigO}[1]{\ensuremath{\mathcal{O}\left( #1 \right)}}

\newcommand{\circEqual}[2]{%
  \begin{tabular}{@{}c@{}c@{}c@{}}
    \begin{tabular}{c}
      #1
    \end{tabular}&
                   \begin{tabular}{c}
                     {\Large=}
                   \end{tabular}&
                                  \begin{tabular}{c}
                                    #2
                                  \end{tabular}%
  \end{tabular}%
}

\begin{document}

\title{Practical Quantum Computing: solving the wave equation using a quantum approach}
\author{Adrien Suau}
\email{adrien.suau@cerfacs.fr}
\affiliation{%
  CERFACS, 42 Avenue Gaspard Coriolis, 31057 Toulouse, France
}
\affiliation{%
  LIRMM, University of Montpellier, 161 rue Ada, 34095 Montpellier, France
}
\author{Gabriel Staffelbach}
\affiliation{%
  CERFACS, 42 Avenue Gaspard Coriolis, 31057 Toulouse, France
}
\author{Henri Calandra}
\affiliation{%
  TOTAL SA, 2 Avenue de Vignancour, 64000 Pau, France
}
\date{\today}

\begin{abstract}
  In the last years, several quantum algorithms that try to address the problem of partial differential equation solving have been devised. On one side, ``direct'' quantum algorithms that aim at encoding the solution of the PDE by executing one large quantum circuit. On the other side, variational algorithms that approximate the solution of the PDE by executing several small quantum circuits and making profit of classical optimisers.
  In this work we propose an experimental study of the costs (in terms of gate number and execution time on a idealised hardware created from realistic gate data) associated with one of the ``direct'' quantum algorithm: the wave equation solver devised in [PCS\@. Costa, S\@. Jordan, A\@. Ostrander, \textit{Phys\@. Rev\@. A} 99, 012323, 2019]. We show that our implementation of the quantum wave equation solver agrees with the theoretical big-O complexity of the algorithm. We also explain in great details the implementation steps and discuss some possibilities of improvements.
  Finally, our implementation proves experimentally that some PDE can be solved on a quantum computer, even if the direct quantum algorithm chosen will require error-corrected quantum chips, which are not believed to be available in the short-term.  
\end{abstract}

\maketitle

\section{\label{sec:introduction}Introduction}

Quantum computing has drawn a lot of attention in the last few years, following the
successive announcements from several world-wide companies about the implementation of
quantum hardware with an increasing number of qubits or reduced error rates~\cite{ibm-quantum-experience,intel-quantum-presskit,google-sycamore-quantum-supremacy,google-quantum-supremacy,2003.01293v1}.

Along with the hardware improvement, new quantum algorithms were discovered, yielding
potential quantum speed-up and applications in various fields such as quantum chemistry~\cite{1812.09976v2}, linear algebra~\cite{Harrow2009,1806.01838v1,PhysRevA.101.022322,1909.03898v1,1909.05820v1,1909.07344v2} or optimisation~\cite{1808.09266v1,1704.04992v3,1411.4028v1}. Recent works even show that differential equations may be solved
by using a quantum computer~\cite{1901.00961v1,1909.06619v1,1809.02622v1,1812.10572v2,1807.04553v1,1701.03684v2,1010.2745v2,0812.4423v1,2002.07868v1,10.1016j.jcp.2020.109347,1907.09032v2}.
But despite the large number of algorithms available, it is hard to find an actual implementation
of a quantum differential equation solver, Hamiltonian simulation being the unique exception
by solving the time-dependant Schr\"odinger equation.

In this work, we present and analyse a quantum wave equation solver we implemented from scratch
according to the algorithm depicted in~\cite{1711.05394v1}. During the solver implementation, we had
to look for a Hamiltonian Simulation procedure. The implementations we found being too
restricted, we decided to implement our own Hamiltonian Simulation procedure, which will also be analysed.

To the best of our knowledge, this work is the first to analyse experimentally the characteristics
of a quantum PDE solver. Such a study has already been performed on the HHL algorithm, in~\cite{1505.06552v2}. 
We checked that the practical implementation agrees with the theoretical asymptotic complexities
on several quantities of interest such as the total gate count with respect to the number of discretisation points used or the precision, the number of qubits required versus the number of discretisation points used to approximate the solution or precision of the solution when compared to a classical finite-difference solver. Finally, we verified that the execution time of the generated quantum circuit on today's accessible quantum hardware was still following the theoretical asymptotic complexities devised for the total gate count. Quantum hardware data were extracted from IBM Q chips.

We show experimentally that it is possible to solve the 1-dimensional wave equation on a quantum computer
with a time-complexity that grows as \bigO{N_d^{3/2}\log\left(N_d\right)^2} where $N_d$ is the number
of discretisation points used to approximate the solution. But even if the asymptotic scaling is better
than classical algorithms, we found out that the constants hidden in the big-O notation were huge
enough to make the solver less efficient than classical solvers for reasonable discretisation sizes.

\section{Problem considered}\label{sec:problem_considered}

We consider a simplified version of the
wave equation on the 1-dimensional line $[0, 1]$ where the propagation speed $c$
is constant and equal to $1$. This equation can be written as

\begin{equation}
  \label{eq:simplified_wave_eq}
  \frac{\partial^2}{\partial t^2} \phi(x,t) = \frac{\partial^2}{\partial x^2}\phi(x,t).
\end{equation}

Moreover, we only consider solving \cref{eq:simplified_wave_eq} with the Dirichlet boundary conditions
\begin{equation}
  \label{eq:dirichlet_bound_cond}
  \frac{\partial}{\partial x} \phi(0, t) = \frac{\partial}{\partial x} \phi(1, t) = 0.
\end{equation}

No assumption is made on initial speed $\phi(x, 0)$ and initial velocity
$\frac{\partial \phi}{\partial t}(x, 0)$.

The resolution of this simplified wave equation on a quantum computer is an
appealing problem for the first implementation of a PDE solver for several
reasons. First, the wave equation is a well-known and intensively studied
problem for which a lot of theoretical results have been verified. Secondly,
even-though it is a relatively simple PDE, the wave equation can be used to
solve some interesting problems such as seismic imaging~\cite{Bamberger1979,bamberger1977application}. Finally, the theoretical
implementation of a quantum wave equation solver has already
been studied in~\cite{1711.05394v1}.

In this paper, we present the complete implementation of a
1-dimensional wave equation solver using quantum technologies based on \texttt{qat} library. To the best of
our knowledge, this work is the first to consider the implementation of an
entire PDE solver that can run on a quantum computer. Specifically, we explain
all the implementation details of the solver from the mathematical theory to the
actual quantum circuit used. The characteristics of the solver are then discussed and
analysed, such as the estimated gate count and estimated execution time on
real quantum hardware. We show that the implementation follows the theoretical
asymptotic behaviours devised in~\cite{1711.05394v1}. Moreover, the wave equation
solver algorithm relies critically on an efficient implementation of a Hamiltonian
simulation algorithm, which we have also implemented and analysed thoroughly.


\section{Implementation}\label{sec:implementation}

The algorithm used to solve the wave equation is explained in~\cite{1711.05394v1}
and uses a Hamiltonian simulation procedure.~\citeauthor{1711.05394v1} chose
the Hamiltonian simulation algorithm described in~\cite{1501.01715v3} for its
nearly optimal theoretical asymptotic behaviour. We privileged instead
the Hamiltonian simulation procedure explained in~\cite{2004-ahokas-graeme-robert-improved-algorithms-for-approximate-quantum-fourier-transforms-and-sparse-hamiltonian-simulations,quant-ph0508139v2}
for its good experimental results based on~\cite{1711.10980v1} and its simpler
implementation (detailed in \cref{sec:pf_impl_details}).

The code has been written using \texttt{qat}, a Python library
shipped with the Quantum Learning Machine (QLM), a package developed and
maintained by Atos. It has not been extensively optimized yet, which means that
there is still a large room for possible improvements.

All the circuits used in this paper have been generated with a subset of
\texttt{qat}'s gate set:
\begin{equation}
  \label{eq:qat_gateset}
  \left\{ H, X, R_y\left( \theta \right), P_h\left( \theta \right), CP_h\left( \theta \right), CNOT, CCNOT \right\}
\end{equation}
and have then been translated to the gate set
\begin{equation}
  \label{eq:ibmq_gateset}
  \left\{ U_1\left( \lambda \right), U_2\left( \lambda, \phi \right), U_3\left( \lambda, \phi, \theta \right), CNOT \right\}
\end{equation}
for $U_1$, $U_2$ and $U_3$ defined in Equation (7) of~\cite{1804.03719v1} as follow:
\begin{equation}
  \label{eq:IBM_base_gates}
  U(\lambda, \phi, \theta) =
  \begin{pmatrix}
    \cos\left( \frac{\theta}{2} \right) & -e^{i\lambda} \sin\left( \frac{\theta}{2} \right) \\
    e^{i\phi} \sin\left( \frac{\theta}{2} \right) & e^{i(\lambda + \phi)} \cos\left( \frac{\theta}{2} \right)
  \end{pmatrix}
\end{equation}
\begin{equation}
  \label{eq:u3}
  U_3(\lambda, \phi, \theta) = U(\lambda, \phi, \theta)
\end{equation}
\begin{equation}
  \label{eq:u2}
  U_2(\lambda, \phi) = U\left(\frac{\pi}{2}, \lambda, \phi\right)
\end{equation}
\begin{equation}
  \label{eq:u1}
  U_1(\lambda) = U(0, 0, \lambda)
\end{equation}
\begin{note} 
  The target gate set presented in \cref{eq:ibmq_gateset}
  does not correspond to the physical gate set implemented by IBM hardware
  (see Equation (8) of~\cite{1804.03719v1}). This choice is justified by the
  fact that IBM only provides hardware characteristics such as gate times
  for the gate set of \cref{eq:ibmq_gateset} and not for the
  real hardware gate set. 
\end{note}

This implementation aims at validating in practice the theoretical asymptotic
complexities of Hamiltonian simulation algorithms and providing a
proof-of-concept showing that it is possible to solve a partial differential
equation on a quantum computer.

\subsection{Sparse Hamiltonian simulation algorithm}\label{sec:sparse_hamsim_res}


\begin{definition}{\textit{$s$-sparse matrix}:}
  A $s$-sparse matrix with $s\in \mathbb{N}^*$ is a matrix that has at most $s$ non-zero entries per row and per column
\end{definition}

\begin{definition}{\textit{sparse matrix}:}
  A sparse matrix is a $s$-sparse matrix with $s \in \bigO{\log(N)}$, $N$ being the size of the matrix.
\end{definition}

In the past years, a lot of algorithms have been devised to simulate the
effect of a Hamiltonian on a quantum state~\cite{1606.02685v2,1610.06546v2,1707.05391v1,1807.03967v1,0910.4157v4,1501.01715v3,quant-ph0508139v2,1412.4687v1,1805.00582v1,1312.1414v2,1202.5822v1,1003.3683v2}.
Among all these algorithms, only few have already been implemented for
specific cases~\cite{qiskit.aqua.evolve,simcount} but to the best of our
knowledge no implementation is currently capable of simulating a generic sparse
Hamiltonian. 

The domain of application of the already existing methods being too
narrow, we decided to implement our own generic sparse Hamiltonian simulation
procedure. We based our work on the product-formula approach described in~\cite{2004-ahokas-graeme-robert-improved-algorithms-for-approximate-quantum-fourier-transforms-and-sparse-hamiltonian-simulations,quant-ph0508139v2}.
One advantage of this approach is that product-formula based algorithms have
already been thoroughly analysed both theoretically~\cite{quant-ph0508139v2,2004-ahokas-graeme-robert-improved-algorithms-for-approximate-quantum-fourier-transforms-and-sparse-hamiltonian-simulations}
and practically~\cite{1711.10980v1,1505.06552v2}, and several
implementations are publicly available, though restricted to Hamiltonians that
can be decomposed as a sum of tensor products of Pauli matrices.
Moreover,~\cite{2004-ahokas-graeme-robert-improved-algorithms-for-approximate-quantum-fourier-transforms-and-sparse-hamiltonian-simulations}
provides a lot of implementation details that allowed us to go straight to the
development step.

Our implementation is capable of simulating an arbitrary sparse Hamiltonian
provided that it has already been decomposed into a sum of $1$-sparse
Hermitian matrices with either only real or only complex entries, each described
by an oracle.
The implementation has been validated with several automated tests and a more
complex case involving the simulation of a $2$-sparse Hamiltonian and described
in \cref{sec:quantum_wave_eq_solver}. Furthermore, it agrees perfectly
with the theoretical complexities devised in~\cite{2004-ahokas-graeme-robert-improved-algorithms-for-approximate-quantum-fourier-transforms-and-sparse-hamiltonian-simulations,quant-ph0508139v2}
as studied and verified in \cref{sec:results}.


\subsection{Quantum wave equation solver}\label{sec:quantum_wave_eq_solver}

Using the Hamiltonian simulation algorithm implementation, we successfully
implemented a 1-dimensional wave equation solver using the algorithm described
in~\cite{1711.05394v1} and explained in
\cref{sec:herm_matrix_construction_decomposition} and \cref{sec:oracle_construction}.

For the specific case considered (\cref{eq:simplified_wave_eq} and
\cref{eq:dirichlet_bound_cond}), 
solving the wave equation for a time $T$ on a quantum computer boils down to
simulating a $2$-sparse Hamiltonian for a time $f(T)$, the function $f$ being
thoroughly described in~\cite{1711.05394v1} and \cref{eq:f_t_def}.
The constructed quantum circuit can then be
applied to a quantum state representing the initial position $\psi(x, 0)$ and
velocity $\frac{\partial \phi}{\partial t} (x, 0)$, and will evolve this state
towards a quantum state representing the final position $\phi(x, T)$
and velocity $\frac{\partial \phi}{\partial t}(x, T)$.

As for the Hamiltonian simulation procedure, the practical results we obtain from the implementation of the quantum wave equation solver seems to match the theoretical asymptotic complexities. See \cref{sec:results} for an analysis of the theoretical asymptotic complexities.

\section{Results}\label{sec:results}

Using a simulator instead of a real quantum computer has several advantages. In
terms of development process, a simulator allows the developer to perform several
actions that are not possible as-is on a quantum processor such as describing a
quantum gate with a unitary matrix instead of a sequence of hardware operations.
Another useful operation that is possible on a quantum simulator and
not currently achievable on a quantum processor is efficient generic state
preparation.

Our implementation uses only standard quantum gates and does not leverage any of
the simulator-only features such as quantum gates implemented from a unitary
matrix. In other words, both the Hamiltonian simulation procedure and the
quantum wave equation solver are ``fully quantum'' and are readily executable on
a quantum processor, provided that it has enough qubits. As a proof, and in
order to benchmark our implementation, we translated the generated quantum circuits
to IBM Q Melbourne gate-set (see \cref{eq:ibmq_gateset}). IBM Q Melbourne~\cite{ibmq-melbourne}
is a quantum chip with 14 usable qubits made available by IBM the 23\textsuperscript{th}
of September, 2018.

\begin{note}

  We chose IBM Q Melbourne mainly because, at the time of writing, it was the publicly accessible quantum chip with the larger number of qubits and so was deemed to be the closest to future quantum hardware. It is important to note that even if IBM Q Melbourne has 14 qubits, the quantum circuits constructed in this paper are not runnable because they require more qubits. Consequently, because of this hardware limitation, hardware topology has also been left apart of the study.
\end{note}

This allowed us to have an estimation of the number of hardware gates needed to either
solve the wave equation or simulate a specific Hamiltonian on this specific hardware.
Combining these numbers and the hardware gate execution time published in~\cite{melbourne.gatetimings}, we were able to compute a rough approximation of
the time needed to solve the considered problem presented in
\cref{eq:simplified_wave_eq} and \cref{eq:dirichlet_bound_cond} on this specific hardware.

\subsection{Hamiltonian simulation}\label{sec:ham_sim_results}

As explained in \cref{sec:sparse_hamsim_res}, the Hamiltonian simulation
algorithm implemented has been first devised in~\cite{2004-ahokas-graeme-robert-improved-algorithms-for-approximate-quantum-fourier-transforms-and-sparse-hamiltonian-simulations,quant-ph0508139v2}.
A quick review of the algorithm along with implementation details can be
found in \cref{sec:pf_impl_details}.
This Hamiltonian simulation procedure requires that the Hamiltonian matrix $H$
to simulate can be decomposed as 
\begin{equation}
  \label{eq:H_decomposition}
  H = \sum_{j=1}^{m} H_j
\end{equation}
where each $H_j$ is an efficiently simulable Hermitian matrix.

In our benchmark, we simulated the Hamiltonian described in \cref{eq:H_matrix}.
According to~\cite{2004-ahokas-graeme-robert-improved-algorithms-for-approximate-quantum-fourier-transforms-and-sparse-hamiltonian-simulations},
real $1$-sparse Hermitian matrices with only $1$ or $0$ entries can be
simulated with $\bigO{n}$ gates and $2$ calls to the oracle, $n$ being the
number of qubits the Hamiltonian $H$ acts on. The exact gate count can be found in \cref{tab:gate_counts} in the row \texttt{$1$-sparse HS}.

Let $O_i$ be the gate complexity of
the oracle implementing the $i$\textsuperscript{th} Hermitian matrix $H_i$ of the decomposition in
\cref{eq:H_decomposition}, we end up with an asymptotic complexity
of $\bigO{n + O_i}$ to simulate $H_i$. Once again, the exact gate count is decomposed in \cref{tab:gate_counts}. 

Applying the Trotter-Suzuki product-formula of order $k$ (see
Definition~\cref{def:lie-trotter-suzuki-prod-formula} in \cref{sec:recomposition_step}
for the definition of the Trotter-Suzuki product-formula) on the quantum circuit
simulating the Hermitian matrices produces a circuit of size

\begin{equation}
  \label{eq:after_pf_trotter_application}
  \bigO{5^k \sum_{i = 1}^m (n + O_i)}.
\end{equation}

This circuit should finally be repeated $r$ times in order to achieve an error
of at most $\epsilon$, with
\begin{equation}
  \label{eq:r_o_notation}
  r \in \bigO{5^{k} m \tau {\left(\frac{m \tau}{\epsilon}\right)}^{\frac{1}{2k}}},
\end{equation}
and $\tau = t \max_i \vert\vert H_i \vert\vert$, $t$ being the time
for which we want to simulate the given Hamiltonian and
$\vert\vert\cdot\vert\vert$ being the spectral norm~\cite{quant-ph0508139v2}.

Merging \cref{eq:after_pf_trotter_application} and
\cref{eq:r_o_notation} gives us the complexity

\begin{equation}
  \label{eq:ham_sim_complexity}
  \bigO{5^{2k} m \tau {\left(\frac{m \tau}{\epsilon}\right)}^\frac{1}{2k} \sum_{i = 1}^m (n + O_i)}.
\end{equation}

This generic expression of the asymptotic complexity can be specialized to our
benchmark case. The number of gates needed to implement the
oracles is $\bigO{n^2}$ and the chosen decomposition contains $m = 2$
Hermitian matrices, each with a spectral norm of $1$. Replacing the symbols in
\cref{eq:after_pf_trotter_application} and \cref{eq:r_o_notation}
results in the asymptotic gate complexity of
\begin{equation}
  \label{eq:oracle_complexity_hamsim}
  \bigO{%
    5^{k} n^2%
  }
\end{equation}
for the circuit
simulating $e^{-iHt/r}$ and a number
\begin{equation}
  \label{eq:repetition_complexity_hamsim}
  r \in \bigO{5^k t {\left( \frac{t}{\epsilon} \right)}^{\frac{1}{2k}}}
\end{equation}
of repetitions, which lead to a total gate complexity of
\begin{equation}
  \label{eq:ham_sim_specialised_complexity}
  \bigO{%
    5^{2k} n^2 t {\left( \frac{t}{\epsilon} \right)}^{\frac{1}{2k}}%
  }.
\end{equation}

In order to check that our implementation follows this theoretical asymptotic
behaviour, we chose to let $k = 1$ and plotted the number of gates generated
versus the three parameters that have an impact on the number of gates: the
number of discretisation points $N_d$ (\cref{fig:gate_num_vs_discr_size}),
the time of simulation $t$ (\cref{fig:gate_num_vs_sim_time}) and the
precision $\epsilon$ (\cref{fig:gate_num_vs_precision}). The corresponding
asymptotic complexity should be
\begin{equation}
  \label{eq:ham_sim_specialised_complexity_k}
  \bigO{%
    n^2 \frac{t^{3/2}}{\sqrt{\epsilon}}%
  } = \bigO{%
    \log_2{\left( N_d \right)}^2 \frac{t^{3/2}}{\sqrt{\epsilon}}%
  }.
\end{equation}

A small discrepency can be observed in \cref{fig:gate_num_vs_discr_size}: the theoretical asymptotic number of gates is \bigO{\log_2\left(N\right)^2} but the experimental values seem better fitted with an asymptotic behaviour of \bigO{\log_2\left(N\right)^{7/4}}. This may be caused by the asymptotic regime not being reached yet.

\begin{figure}
  \centering
  \subfigure[{%
    Number of quantum gates versus simulated matrix size. The values of the constants $\gamma = 250\,000$ and $\gamma_0 = 2\,000\,000$ have been chosen arbitrarily to fit the experimental data.%
  }]{\includegraphics[width=.47\linewidth]{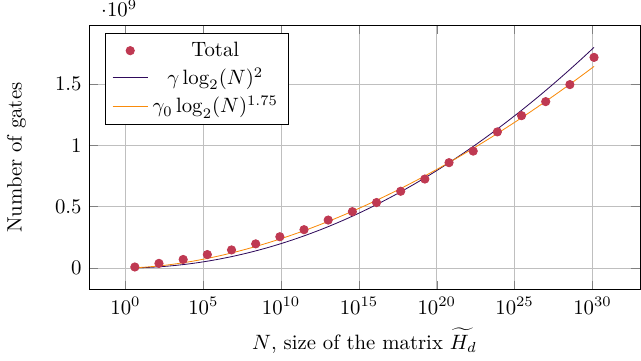}\label{fig:gate_num_vs_discr_size}}\hspace{.03\linewidth}
  \subfigure[{%
    Number of quantum gates versus physical time. The value of $\beta = 39\,000\,000$ has been chosen arbitrarily to fit the experimental data.%
  }]{\includegraphics[width=.47\linewidth]{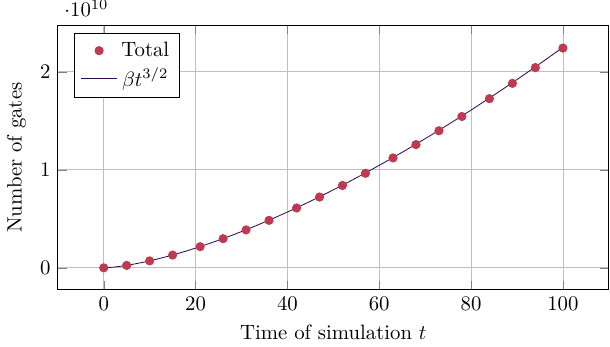}\label{fig:gate_num_vs_sim_time}}
  \subfigure[{%
    Number of quantum gates versus targeted solution precision. The value of $\alpha = 130\,000$ has been chosen arbitrarily to fit the experimental data.%
  }]{\includegraphics[width=.7\linewidth]{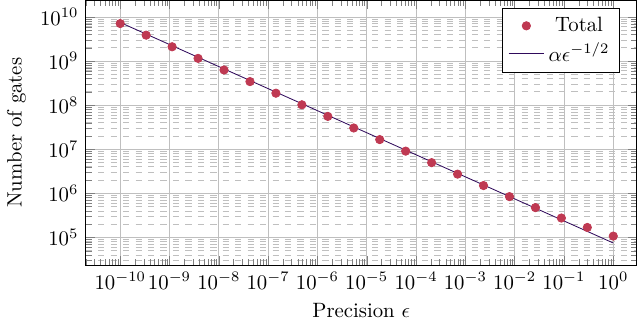}\label{fig:gate_num_vs_precision}}
  \caption{Number of quantum gates needed to simulate the
    Hamiltonian described in \cref{sec:herm_matrix_construction_decomposition}
    using the oracles implemented following \cref{sec:oracle_construction}.
    Graphs generated with a Trotter-Suzuki product-formula order $k = 1$,
    $32$ discretisation points (i.e.\ $n = 6$ qubits) for \cref{fig:gate_num_vs_sim_time}
    and \cref{fig:gate_num_vs_precision},
    a physical time $t = 1$ for \cref{fig:gate_num_vs_discr_size} and
    \cref{fig:gate_num_vs_precision} and a precision $\epsilon =
    10^{-5}$ for \cref{fig:gate_num_vs_discr_size} and \cref{fig:gate_num_vs_sim_time}.}
\end{figure}


\subsection{Wave equation solver}\label{sec:wave_eq_solver_results}

\begin{figure}
  \centering
  \includegraphics[width=.7\linewidth]{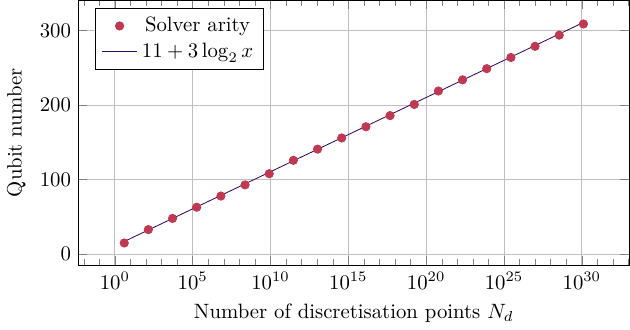}
  \caption{Plot of the number of logical qubits needed to run the wave equation solver
    for a time $t=1$, a precision $\epsilon = 10^{-5}$ and a Trotter-Suzuki
    product-formula of order $k = 1$. The constants values $11$ and $3$ have been
    chosen arbitrarily to fit the experimental data. The number of physical qubits needed
    will depend on their error rate as noted in~\cite{Fowler2012}. Multiplying the number of
    logical qubits by $3$ to $4$ orders of magnitude might be a good estimate of the actual
    number of physical qubits required.}\label{fig:arity_vs_discr_size}
\end{figure}

The first characteristic of the wave equation solver that needs to be checked is its validity: is
the quantum wave equation solver capable of solving accurately the wave equation as described in
\cref{eq:simplified_wave_eq} and \cref{eq:dirichlet_bound_cond}?

To check the validity of the solver, we used \texttt{qat} simulators and Atos QLM to simulate the
quantum program generated to solve the wave equation with different values for the number of discretisation
points $N_d$, for the physical time $t$ and for the precision $\epsilon$.
\cref{fig:classical-vs-quantum-with-error} shows the classical solution
versus the quantum solution and the absolute error between the two solutions
for $N_d = 32$, $t = 0.4$ and $\epsilon = 10^{-3}$.
The solution obtained by the quantum solver is nearly exactly the same as the classical solution
obtained with finite differences. The error between the two solutions is of the order of $10^{-7}$,
which is $4$ orders of magnitudes smaller than the error we asked for. 

\begin{figure}
  \centering
  \subfigure[{%
    Quantum versus classical solution. Solutions are not visually distinguishable on the graph, see the associated absolute error.%
  }]{\includegraphics[width=.47\linewidth]{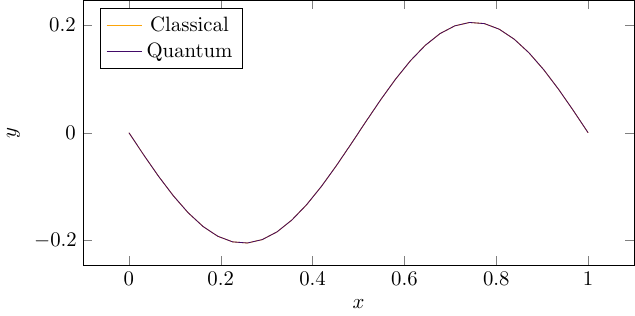}\label{fig:classical_vs_quantum}}\hspace{.03\linewidth}%
  \subfigure[{%
    Absolute error between the solution obtained by a classical finite-difference solver and the solution computed with the quantum solver.%
  }]{\includegraphics[width=.47\linewidth]{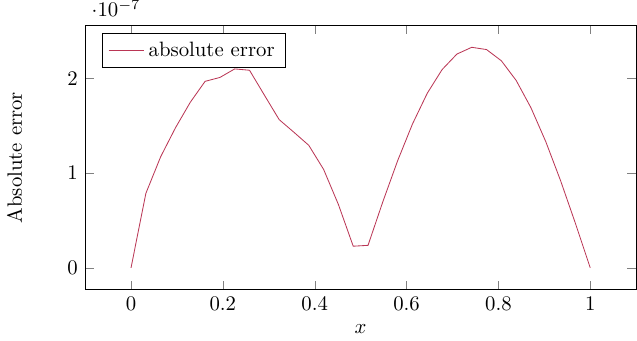}\label{fig:classical_vs_quantum_error}}
  \caption{Comparison of the classical solver and the quantum solver. Both solvers solved
    the $1$-D wave equation with $N_d = 32$ discretisation points and a physical time of
    $t = 0.4$. The classical solver uses finite-differences with a very small time-step
    in order to avoid as much as possible errors due to time-discretisation. The quantum
    solver was instructed to solve the wave equation with a precision of at least
    $\epsilon = 10^{-3}$, used a Trotter-Suzuki order of $k = 1$.
    The solutions of the two solvers are too close to be able to
    notice a difference (they overlap on the graph), that is why a second graph plotting the
    absolute error between the two solvers is included.}\label{fig:classical-vs-quantum-with-error}
\end{figure}


Once the validity of our solver has been checked on multiple test cases, the next interesting
property we would like to verify is the asymptotic cost: does the implemented simulator seem
to agree with the theoretical asymptotic complexities derived from~\cite{1711.05394v1}
and~\cite{quant-ph0508139v2}?

In our specific case, the Hamiltonian $H$ to simulate can be decomposed in two
$1$-sparse Hermitian matrices, both of them having a spectral norm of $1$. The exact decomposition can be found in \cref{sec:matrices_construction}. We chose to
let the product-formula order be equal to $k=1$ and reuse the asymptotic complexity
found in \cref{eq:ham_sim_specialised_complexity} by changing the time
of simulation $t$ by the time $f(t)$:

\begin{equation}
  \label{eq:wave_eq_solver_f_complexity}
  \bigO{%
    5^{2k} n^2 f(t) {\left( \frac{f(t)}{\epsilon} \right)}^{\frac{1}{2k}}%
  }.
\end{equation}

Following the study performed in~\cite{1711.05394v1},
\begin{equation}
  \label{eq:f_t_def}
  f(t) = \frac{t}{\delta x} = t \left(N_d - 1\right)
\end{equation}
where $\delta x$ is the distance between two discretisation points.
Moreover, it is possible to prove (see \cref{sec:matrices_construction})
that 
\begin{equation}
  \label{eq:replace_n_with_nd}
  n = \left\lceil \log_2(2N_d - 1) \right\rceil
\end{equation}

Replacing $f(t)$ and $n$ in \cref{eq:after_pf_trotter_application} and
\cref{eq:r_o_notation} gives us a gate complexity of
\begin{equation}
  \label{eq:oracle_complexity_waveeq}
  \bigO{5^k \log_2\left( N_d \right)^2}
\end{equation}
to construct a circuit simulating $e^{-iHt/r}$ and a number of repetitions
\begin{equation}
  \label{eq:repetition_complexity_waveeq}
  r \in \bigO{5^k t N_d {\left( \frac{t N_d}{\epsilon} \right)}^{\frac{1}{2k}}}.
\end{equation}

Merging the two expression results in a gate complexity of

\begin{equation}
  \label{eq:wave_eq_solver_complexity}
  \bigO{%
    5^{2k} t N_d {\log_2(N_d)}^2 {\left(\frac{t N_d}{\epsilon}\right)}^{\frac{1}{2k}}%
  }.
\end{equation}

Choosing the Totter-Suzuki formula order $k = 1$ gives us a final complexity of
\begin{equation}
  \label{eq:wave_eq_solver_complexity_k}
  \bigO{%
    N_d^{3/2} {\log_2(N_d)}^2 \frac{t^{3/2}}{\sqrt{\epsilon}}%
  }
\end{equation}
to solve the wave equation presented in \cref{eq:simplified_wave_eq}.
This theoretical result is verified experimentally in
\cref{fig:gate_num_vs_discr_size_wave_eq_solver}.

\begin{figure}
  \centering
  \subfigure[{%
    Number of quantum gates needed to solve the wave equation described in \cref{eq:simplified_wave_eq} versus discretisation size. The  value of $\lambda = 300\,000$ has been chosen arbitrarily to fit the experimental data.%
  }]{\includegraphics[width=.47\linewidth]{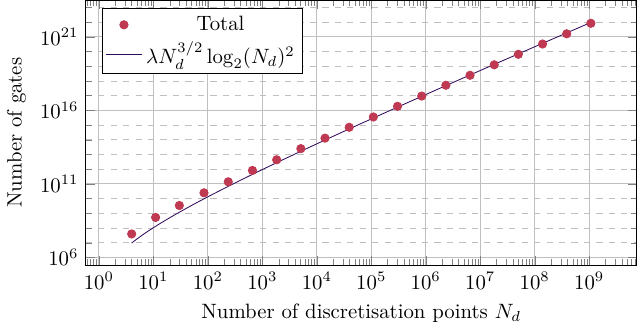}\label{fig:gate_num_vs_discr_size_wave_eq_solver}}\hspace{.03\linewidth}%
  \subfigure[{%
    Estimated execution time of the wave equation solver on IBM Q Melbourne hardware. Individual gate times have been extracted from~\cite{melbourne.gatedecomposition} and~\cite{melbourne.gatetimings}. GF pulse time has been approximated via arithmetic mean to $347$ns, GD pulse time is $100$ns and buffer time is $20$ns. The value of $\delta = 0.06$ has been chosen arbitrarily to fit the experimental data.%
  }]{\includegraphics[width=.47\linewidth]{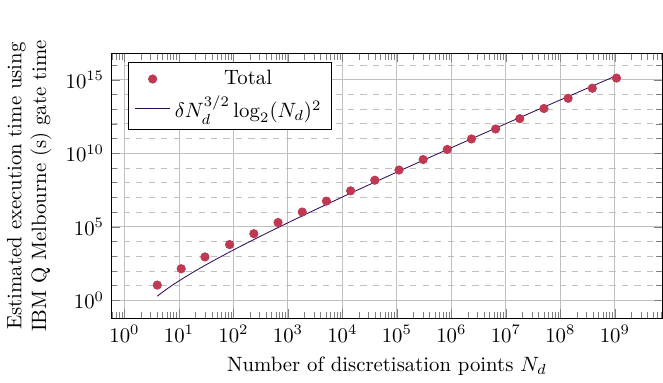}\label{fig:execution_time_melbourne_vs_discr_size}}
  \caption{Graphs generated with a
    Trotter-Suzuki product-formula order $k = 1$, a physical time $t=1$ and a
    precision $\epsilon = 10^{-5}$. }
\end{figure}

\section{Discussion}\label{sec:discussion}

In this work, we focus on the practical cost of implementing a $1$-dimensional quantum wave equation solver on a quantum computer. We show that a quantum computer is able to solve partial differential equations by constructing and simulating the quantum circuits described. We also study the scaling of the solver with respect to several parameters of interest and show that the theoretical asymptotic bounds are mostly verified.

In future works, one can study the possibilities of circuit optimisation. It would also be interesting to implement Neumann boundary conditions instead of Dirichlet ones. A practical implementation including a non-constant propagation speed $c$ has also been realised during the writing of this paper. The results were encouraging but were not judged mature enought to include them in the paper. Finally, future works might want to extend the wave equation solver to $2$ dimensions or more. 



\begin{acknowledgements}
  The authors would like to thank Reims University, the ROMEO HPC center, Total, the CCRT and Atos for their support by giving us access to Atos quantum simulator.
\end{acknowledgements}

\section*{Supplementary material}
\label{sec:suppl-mater}

The implementation of the quantum wave equation solver is available at \url{https://gitlab.com/cerfacs/qaths}.
The \texttt{qprof} tool is available at \url{https://gitlab.com/qcomputing/qprof/qprof}.

\bibliographystyle{ACM-Reference-Format}
\bibliography{quantum,quantum-url}

\appendix

\section{Product-formula implementation details}\label{sec:pf_impl_details}

\subsection{Hamiltonian simulation}\label{sec:hamiltonian_simulation_definition}

Hamiltonian simulation is the problem of constructing a quantum circuit that will
evolve a quantum state according to a Hamiltonian matrix, following the
Schrödinger equation. In other words, Hamiltonian simulation algorithms generate a
quantum circuit performing the unitary transformation $U$ such that $ || U - e^{-iHt} || < \epsilon$,
$H$ being a given Hamiltonian matrix, $t$ a time of evolution and $\epsilon$ a
precision with respect to $||\cdot{}||$, the spectral norm.

Several quantum algorithms have been developed in the last few years to solve the problem
of $s$-sparse Hamiltonian simulation~\cite{1606.02685v2,1610.06546v2,1707.05391v1,1807.03967v1,0910.4157v4,1501.01715v3,quant-ph0508139v2,1412.4687v1,1805.00582v1,1312.1414v2,1202.5822v1,1003.3683v2}.
Among these algorithms we decided to implement the product-formula approach~\cite{quant-ph0508139v2,2004-ahokas-graeme-robert-improved-algorithms-for-approximate-quantum-fourier-transforms-and-sparse-hamiltonian-simulations},
for the reasons presented in \cref{sec:sparse_hamsim_res}.

The product formula algorithm has three main steps: decompose, simulate,
recompose. It works by first decomposing the $s$-sparse Hamiltonian matrix $H$ that
should be simulated as a sum of Hermitian matrices $H_j$ that are considered
easy to simulate
\begin{equation}
  \label{eq:H_sum_Hj}
  H = \sum_{j = 0}^{m-1} H_j.
\end{equation}
The second step is then to simulate each $H_j$ separately, i.e.\ to create
quantum circuits implementing $e^{-iH_{j}t}$ for all the $H_j$ in the decomposition
in \cref{eq:H_sum_Hj}. The last step uses the simulations computed in
step two to approximate $e^{-iHt}$.

The very first questions that should be answered before starting any
implementation of the product-formula algorithm are ``What is an easy to
simulate matrix?'' and ``What kind of Hermitian matrices are easy to
simulate?''.

\subsection{Easy to simulate matrices}\label{sec:easy_simu_matrices}

One of the most desirable properties for an ``easy to simulate'' matrix is the
possibility to simulate it exactly, i.e.\ to construct a quantum circuit that
will perfectly implement $e^{-iHt}$. This property becomes a requirement when one
wants rigorous bounds on the error of the final simulation. Another enviable
property of these matrices is that they can be simulated with
a low gate number and only a few calls to the matrix oracle. 

\begin{definition}[Easy to simulate matrix]
  \label{def:easy-to-simulate}
  A Hermitian matrix $H$ can be qualified as ``easy to simulate'' if there exist
  an algorithm that takes as input a time $t$ and the matrix $H$ and outputs a quantum
  circuit ${C(H)}_t$ such that
  \begin{enumerate}
  \item The quantum circuit ${C(H)}_t$ implements exactly the unitary transformation
    $e^{-iHt}$, i.e.
    \begin{equation*}
      \label{eq:exact-implementation-hamsim}
      \left| \left| e^{-iHt} - {C(H)}_t \right| \right| = 0.
    \end{equation*}
  \item The algorithm only needs $\bigO{1}$ calls to the oracle of $H$ and
    $\bigO{\log N}$ additional gates, $N$ being the dimension of the matrix $H$.
  \end{enumerate}
\end{definition}


With this definition of an ``easy to simulate'' matrix, we can now search for
matrices or group of matrices that satisfy this definition.

\subsubsection{\label{sec:multiples-identity}Multiples of the identity}

The first and easiest matrices that fulfil the easy to simulate matrix requirements
are the multiples of the identity matrix $\left\{ \alpha I, \, \alpha \in \mathbb{R} \right\}$
with $I$ the identity matrix. The quantum circuit to simulate this class of matrices can be
found in~\cite{2018-almudena-carrera-vazquez-quantum-algorithm-for-solving-tri-diagonal-linear-systems-of-equations}.

\subsubsection{\label{sec:1_sparse_hermitian_simu}$1$-sparse Hermitian matrices}

A larger class of matrices that can be efficiently and exactly simulated are the
$1$-sparse, integer weighted, Hermitian matrices. Quantum circuits simulating exactly $1$-sparse
matrices with integer weights can be found in~\cite{2004-ahokas-graeme-robert-improved-algorithms-for-approximate-quantum-fourier-transforms-and-sparse-hamiltonian-simulations}.

\begin{note}
  Procedures simulating $1$-sparse matrices with real (non-integers) weights are
  also described in the paper, but these matrices do not fall in the ``easy to simulate''
  category because the procedures explained are exact only if all the
  matrix weights can be represented exactly with a fixed-point representation,
  which is not always verified.
\end{note}

\begin{note}
  Multiples of identity matrices presented in 
  \cref{sec:multiples-identity} are a special case of $1$-sparse
  matrices. The two classes have been separated because more efficient
  quantum circuits exists for $\alpha I$ matrices.
\end{note}

\subsection{Decomposition of $H$}\label{sec:H_decomposition}

Once the set of ``easy to simulate'' matrices has been established, the next
step of the algorithm is to decompose the $s$-sparse matrix $H$ as a sum of
matrices in this set.

There are two possible ways of performing this decomposition, each one with its
advantages and drawbacks: applying a procedure computing the decomposition
automatically, or decompose the matrix $H$ beforehand and provide
the decomposition to the algorithm.

The first solution, which is to automatically construct the oracles of
the $H_j$ matrices from the oracle of the $H$ matrix has been studied in~\cite{2004-ahokas-graeme-robert-improved-algorithms-for-approximate-quantum-fourier-transforms-and-sparse-hamiltonian-simulations}
and~\cite{1003.3683v2}. Thanks to this automatic decomposition procedure, we only need to implement one oracle. This simplicity comes at the cost of a higher gate count: each call to the automatically constructed oracles of the matrices $H_j$ will require several calls to the oracle of $H$ along with additional gates.

On the other hand, the second solution offers more control at the cost of less abstraction and more work. The decomposition of $H$ is not automatically computed and should be performed beforehand. Once the matrix $H$ has been decomposed as in \cref{eq:H_decomposition}, the oracles for the matrices $H_j$ should be implemented.
This means that we should now implement $m$ oracles instead of only $1$ for the first solution. The main advantage of this method over the one using automatic-decomposition is that it gives us more control, a control that can be used to optimize even more the decomposition of \cref{eq:H_sum_Hj} (less $H_j$ in the decomposition, $H_j$ matrices that can be simulated more efficiently, \ldots).


All the advantages and drawbacks weighted, we chose to implement the second option for
several reasons. First, the implementation of the automatic decomposition
procedure adds a non-negligible implementation complexity to the whole
Hamiltonian simulation procedure. Moreover, the automatic decomposition
procedure can be implemented afterwards and plugged effortlessly to the
non-automatic implementation. Finally, our use-case only required to simulate
a $2$-sparse Hamiltonian that can be decomposed as the sum of two $1$-sparse, easy to simulate,
Hermitian matrices, which makes the manual decomposition step manageable. 

\subsection{Simulation of the $H_j$}\label{sec:simulation_of_Hj}

Once the matrix $H$ has been decomposed following \cref{eq:H_sum_Hj}
with each $H_j$ being an ``easy to simulate'' matrix, the simulation
of $H_j$ becomes a straightforward application of the procedures described in
\cref{sec:easy_simu_matrices}.

After this step, we have access to quantum circuits implementing $e^{-iH_{j}t}$ for
$j \in \left[ 0, m-1 \right]$ and $t\in \mathbb{R}$.

\subsection{Re-composition of the $e^{-iH_{j}t}$}\label{sec:recomposition_step}

The ultimate step of the algorithm is to approximate the desired evolution $e^{-iHt}$ with the
evolutions $e^{-iH_{j}t}$. In the special case of mutually commuting $H_j$, this step is trivial
as it boils down to use the properties of the exponential function on matrices and write
$e^{iHt} = e^{i\sum_{j}H_{j}t} = \prod_j e^{iH_{j}t}$. But in the more realistic case where the matrices
$H_j$ do not commute, a more sophisticated method should be used to approximate the evolution
$e^{-iHt}$. To this end, we used the first-order Lie-Trotter-Suzuki product formula defined in
Definition~\cref{def:lie-trotter-suzuki-prod-formula}.

\begin{definition}[%
  Lie-Trotter-Suzuki product formula~\cite{Suzuki1990,Suzuki1986,1711.10980v1}%
  ]\label{def:lie-trotter-suzuki-prod-formula}
  The Lie-Trotter-Suzuki product formula approximates
  \begin{equation}
    \label{eq:exp-of-sum}
    \exp \left( \lambda \sum_{j=0}^{m-1} \alpha_j H_j  \right)
  \end{equation}
  with 
  \begin{equation}
    \label{eq:lie-trotter-suzuki-product-formula}
    S_2(\lambda) = \prod_{j=0}^{m-1} e^{\alpha_j H_j \lambda / 2}  \prod_{j=m-1}^{0} e^{\alpha_j H_j \lambda / 2}
  \end{equation}
  and can be generalized recursively to higher-orders
  \begin{equation}
    \label{eq:lie-trotter-suzuki-product-formula-order-k}
      S_{2k}\left( \lambda \right) =  {\left[ S_{2k-2}\left( p_k \lambda \right) \right]}^2 \times S_{2k-1}\left( \left( 1 - 4p_k \right) \lambda \right) {\left[ S_{2k-2}\left( p_k \lambda \right) \right]}^2
  \end{equation}
  with $p_k = {\left( 4 - 4^{1/(2k-1)} \right)}^{-1}$ for $k> 1$.
  Using this formula, we have the approximation
  \begin{equation}
    \label{eq:lie-trotter-suzuki-product-formula-timestep}
    e^{\lambda H} = {\left[ S_{2k}\left( \frac{\lambda}{n} \right) \right]}^n + \bigO{\frac{\vert\lambda\vert^{2k+1}}{n^{2k}}}.
  \end{equation}
  
\end{definition}

We used the Lie-Trotter-Suzuki product formula with $\lambda = -it$ to approximate
the operator $e^{-iHt}$ up to an error of $\epsilon \in \bigO{\frac{t^{2k+1}}{n^{2k}}}$.

\section{Hermitian matrix construction and decomposition}\label{sec:herm_matrix_construction_decomposition}

One of the main challenge in implementing a quantum wave equation solver lies in
the construction and implementation of the needed oracles. This appendix
describes the first step of the implementation process: the construction and
decomposition of the Hamiltonian matrix that will be simulated using the
Hamiltonian simulation procedure introduced in \cref{sec:pf_impl_details}.

This appendix follows the analysis performed in~\cite{1711.05394v1} and adds
details and observations that will be refereed to in
\cref{sec:oracle_construction} when dealing with the actual oracle implementation.

\subsection{Hamiltonian matrix description}\label{sec:hamiltonian_descr}

In order to devise the Hamiltonian matrix that should be simulated to solve the
wave equation, the first step is to discretise \cref{eq:simplified_wave_eq}
with respect to space. Such a discretisation can be seen as a graph $G_{\delta
  x}$ whose vertices are the discretisation points and with edges between
nearest neighbour vertices. The graph $G_{\delta x}$ is depicted in 
\cref{fig:discretisation_graph}.

\begin{figure}
  \centering
  \includegraphics[width=\linewidth]{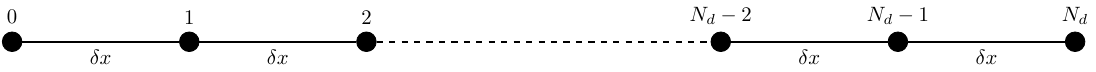}
  \caption{Graph $G_{\delta x}$ built from the discretisation of the
    $1$-dimensional line $[0, 1]$ with $N_d$ discretisation points (i.e.\ $\delta
    x = \frac{1}{N_d-1}$).}\label{fig:discretisation_graph}
\end{figure}

The graph Laplacian of $G_{\delta x}$, defined as
\begin{equation}
  {L(G_{\delta x})}_{i,j} :=
  \begin{cases}
    \deg(v_i) & \text{if } i = j \\
    -1 & \text{if } (i \neq j) \land (v_i \mbox{ adjacent to } v_j) \\
    0 & \text{otherwise} \\
  \end{cases}
  \label{eq:graph_laplacian}
\end{equation}
can then be used to approximate the differential operator
$\frac{\partial^2}{\partial x^2}$. By using the discretisation approximation
\begin{equation}
  \frac{\partial^2 \phi}{\partial x^2}(i \delta x, t) \approx
  \frac{\phi_{i-1,t} - 2 \phi_{i,t} + \phi_{i+1,t}}{{\delta x}^2}
  \label{eq:discr_approx}
\end{equation}
with $\phi_{i,t} = \phi(i\delta x, t)$, 
and approximating $\phi(x, t)$ with a vector $\phi = {\left[ \phi_{i,t}
  \right]}_{0 \leqslant i < N_d}$, the matrix
\begin{equation}
  \label{eq:discr_matrix}
  A = - \frac{1}{{\delta x}^2} L(G_{\delta x})
\end{equation}
approximates the second derivative of $\phi$ when $\delta x \to 0$ as
\begin{equation}
  {\left[ A\phi \right]}_i = \frac{\phi((i-1)\delta x, t) - 2 \phi(i\delta x, t) + \phi((i+1)\delta x, t)}{{\delta x}^2}.
\end{equation}

The approximation in \cref{eq:discr_matrix} is then used in
\cref{eq:simplified_wave_eq} to approximate the spatial derivative operator: 
\begin{equation}
  \label{eq:discr_wave_equation}
  \frac{\partial^2}{\partial t^2}\phi = -\frac{1}{{\delta x}^2} L(G_{\delta x}) \phi.
\end{equation}

Based on this formula,~\cite{1711.05394v1} shows that simulating
\begin{equation}
  \label{eq:H_definition}
  H = 
  \begin{pmatrix}
    0 & B \\
    B^\dagger & 0\\
  \end{pmatrix}
\end{equation}
with
\begin{equation}
  BB^\dagger = L(G_{\delta x})
  \label{eq:bbdagger_l}
\end{equation}
constructs a quantum circuit that will evolve a part of
the quantum state it is applied on according to the discretised wave equation in
\cref{eq:discr_wave_equation}.

A matrix $B$ satisfying \cref{eq:bbdagger_l} can be obtained directly from the
graph $G_{\delta x}$ representing the discretisation. The algorithm to construct the
matrix $B$ can be decomposed in three steps. First, the vertices (discretisation
points) should be arbitrarily ordered by assigning them a unique index in $[0, N_d-1]$.
Then, each edge of the graph is arbitrarily oriented and indexed with indices in $[0,
N_d-2]$. Finally, $B$ is computed with the following definition 
\begin{equation}
  \label{eq:B_definition}
  B_{ij} =
  \left\{
    \begin{array}{ll}
      1  & \text{if edge } j \text{ is a self-loop of vertex } i \text{,}\\
      1  & \text{if edge } j \text{ has vertex } i \text{ as source,}\\
      -1 & \text{if edge } j \text{ has vertex } i \text{ as sink,} \\
      0  & \text{otherwise} \\
    \end{array}
  \right..
\end{equation}

Note that edges' orientation and vertices/edges ordering is completely arbitrary.
Changing either the edges orientation on one of the orderings will change the
matrix $B$ but will not affect $BB^\dagger$ which should be equal to $L(G_{\delta x})$. This
freedom in the ordering and orientation choices takes a crucial importance in
the oracle implementation as it allows us to pick the ordering/orientation that
will produce an easy-to-implement matrix $B$.

\subsection{Dirichlet boundary conditions}\label{sec:dirichlet_boundary}

Fixing boundary conditions is a requirement for most of the partial differential
equations to admit a unique well-defined solution. There exist several boundary
conditions such as Neumann, Dirichlet, Robin or Cauchy ones. For simplicity, we
restricted ourselves to the study of \cref{eq:simplified_wave_eq} with
Dirichlet boundary condition of \cref{eq:dirichlet_bound_cond}.

In the case of Dirichlet boundary conditions on the $1$-dimensional line $[0,
1]$, the two boundary nodes at $x = 0$ and $x = 1$ can be ignored as their value
is always equal to $0$. Moreover,~\cite{1711.05394v1} shows that the graph
$G_{\delta x}^D$ representing the discretisation with Dirichlet boundary
conditions of \cref{eq:dirichlet_bound_cond} is simply $G_{\delta x}$ with
self-loops on the two outer nodes (i.e.\ the ones indexed $1$ and $N_d-2$ as $0$
and $N_d-1$ are ignored). $G_{\delta x}^D$ is depicted in 
\cref{fig:dirichlet_discr_graph}. The algorithm to construct the matrix $B$ remain
the same as explained in \cref{sec:hamiltonian_descr}.

\begin{figure}
  \centering
  \includegraphics[width=\linewidth]{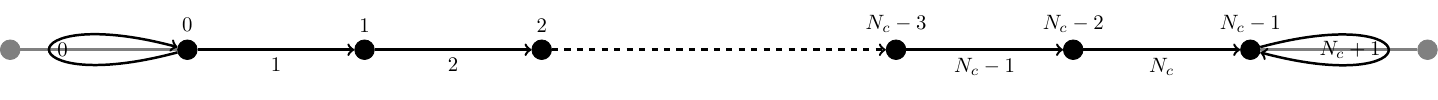}
  \caption{Graph $G_{\delta x}^D$ representing the discretisation of the
    1-dimensional line $[0,1]$ with Dirichlet boundary conditions. The points
    and edges in grey are only drawn for illustration purpose and are ignored in
    the analysis because the boundary condition impose a value of $0$ on these
    vertices. Loops are added to $G_{\delta x}$ to encode the fact that this
    graph represents Dirichlet boundary conditions. Vertices (resp.\ edges) are
    ordered with indices within $[0, N_c-1]$ (resp. $[0, N_c+1]$). $N_c$ is the
    number of considered points and is equal to $N_d - 2$ (the two extremal
    points are ignored).}\label{fig:dirichlet_discr_graph}
\end{figure}

\subsection{Matrices construction}\label{sec:matrices_construction}

All the pieces are now in place to start building the matrix
$B_{\text{d}} \in \mathbb{R}^{ (N_c-1) \times N_c}$. Using the definition of the
matrix $B$ written in \cref{eq:B_definition}
and the graph $G_{\delta x}^D$ depicted in 
\cref{fig:dirichlet_discr_graph} we end up with 

\begin{equation}
  B_{\text{d}} =
  \begin{pmatrix}%
    1      & 1      & 0      & \cdots & 0      \\%
    0      & -1     & 1      & \ddots & \vdots \\%
    \vdots & \ddots & \ddots & \ddots & 0      \\%
    0      & \cdots & 0      & -1     & 1      \\%
  \end{pmatrix}%
  \label{eq:B_matrix}
\end{equation}

We can easily check that $B_{\text{d}}B_{\text{d}}^\dagger$ is equal to the
well-known discretisation matrix
\begin{equation}
  B_{\text{d}}B_{\text{d}}^\dagger =
  \begin{pmatrix}%
    2      & -1     & 0      & \cdots & 0      \\%
    -1     & 2      & \ddots & \ddots & \vdots \\%
    0      & \ddots & \ddots & \ddots & 0      \\%
    \vdots & \ddots & \ddots & 2      & -1     \\%
    0      & \cdots & 0      & -1     & 2      \\%
  \end{pmatrix},%
  \label{eq:BBdag_matrix}
\end{equation}
which validate the method of construction of $B_{\text{d}}$.

Computing $\widetilde{H_{\text{d}}}$, the Hamiltonian matrix that should be simulated to
evolve the quantum state according to the wave equation in
\cref{eq:simplified_wave_eq} with Dirichlet boundary conditions, is now
straightforward. Using \cref{eq:H_definition}, we directly obtain

\begin{equation}
  \label{eq:H_matrix}
  \widetilde{H_d} = \frac{1}{\delta x}
  \begin{pmatrix}
    0      & \cdots & \cdots & 0      & 1      & 1      & 0      & \cdots & 0      \\
    \vdots &        &        & \vdots & 0      & -1     & 1      & \ddots & \vdots \\
    \vdots &        &        & \vdots & \vdots & \ddots & \ddots & \ddots & 0      \\
    0      & \cdots & \cdots & 0      & 0      & \cdots & 0      & -1     & 1      \\
    1      & 0      & \cdots & 0      & 0      & \cdots & \cdots & \cdots & 0      \\
    1      & -1     & \ddots & \vdots & \vdots &        &        &        & \vdots \\
    0      & 1      & \ddots & 0      & \vdots &        &        &        & \vdots \\
    \vdots & \ddots & \ddots & -1     & \vdots &        &        &        & \vdots \\
    0      & \cdots & 0      & 1      & 0      & \cdots & \cdots & \cdots & 0      \\
  \end{pmatrix}
\end{equation}


As explained in \cref{sec:pf_impl_details}, the Hamiltonian simulation algorithm
implemented requires that the Hamiltonian to simulate is split as a sum
of $1$-sparse Hermitian matrices. There are a lot of valid decompositions for
the matrix $\widetilde{H_d}$ and we are free to choose the decomposition that will
simplify the most the oracle implementation or reduce the gate complexity.

We made the choice to decompose $B_{\text{d}}$ as two $1$-sparse matrices and
then reflect this decomposition on $\widetilde{H_d}$. Let $B_1$ and $B_{-1}$
defined as
\begin{equation}
  \label{eq:B1_definition}
  B_1 =
  \begin{pmatrix}%
    0      & 1      & 0      & \cdots & 0      \\%
    \vdots & \ddots & 1      & \ddots & \vdots \\%
    \vdots &        & \ddots & \ddots & 0      \\%
    0      & \cdots & \cdots & 0      & 1      \\%
  \end{pmatrix}%
\end{equation}
\begin{equation}
  \label{eq:Bm1_definition}
  B_{-1} =
  \begin{pmatrix}%
    1      & 0      & \cdots & \cdots & 0      \\%
    0      & -1     & \ddots &        & \vdots \\%
    \vdots & \ddots & \ddots & \ddots & \vdots \\%
    0      & \cdots & 0      & -1     & 0      \\%
  \end{pmatrix}%
\end{equation}
we have $B_{\text{d}} = B_1 + B_{-1}$. Let also
\begin{equation}
  \label{eq:Hi_tilde_definition}
  \widetilde{H_1} = \frac{1}{\delta x}
  \begin{pmatrix}
    0 & B_1 \\
    {B_1}^\dagger & 0\\
  \end{pmatrix}
  ,\quad
  \widetilde{H_{-1}} = \frac{1}{\delta x}
  \begin{pmatrix}
    0 & B_{-1} \\
    {B_{-1}}^\dagger & 0\\
  \end{pmatrix},
\end{equation}
it is easy to see that $\widetilde{H_d} = \widetilde{H_1} + \widetilde{H_{-1}}$
and that both $\widetilde{H_1}$ and $\widetilde{H_{-1}}$ are $1$-sparse Hermitian
matrices.

For convenience, we also define
\begin{equation}
  \label{eq:Hi_definition}
  H_1 = 
  \begin{pmatrix}
    0 & B_1 \\
    {B_1}^\dagger & 0\\
  \end{pmatrix}
  ,\quad
  H_{-1} =
  \begin{pmatrix}
    0 & B_{-1} \\
    {B_{-1}}^\dagger & 0\\
  \end{pmatrix},
\end{equation}
and $H_d = H_1 + H_{-1}$, the $\widetilde{H_1}$, $\widetilde{H_{-1}}$ and $\widetilde{H_d}$ matrices rescaled to contain only integer weights. These matrices have the interesting property that simulating  $\widetilde{H_d}$ (resp. $\widetilde{H_1}$, $\widetilde{H_{-1}}$) for a time $t$ is equivalent to simulating $H_d$ (resp. $H_1$, $H_{-1}$) for a time $\frac{t}{\delta x}$. This property will be used in the following sections as it offers us the opportunity to simulate the integer-weighted matrices  $H_d$, $H_1$ and $H_{-1}$ instead of the real-weighted ones $\widetilde{H_d}$, $\widetilde{H_1}$ and $\widetilde{H_{-1}}$.

Note also that a lower bound of the number of qubits needed to solve the wave equation for $N_d$
discretisation points can be computed from the dimensions of $H_d$. As
the non-empty upper-left block of matrix $H_d$ is of dimension $\left( 2N_d - 1
\right) \times \left( 2N_d - 1 \right)$, we need at least 
\begin{equation}
  \label{eq:qubit_needed_wave_eq}
  \left\lceil \log_2(2 N_d - 1) \right\rceil
\end{equation}
qubits to simulate it. This estimation does not take into account ancilla qubits that may be needed to implement the oracles.

\section{Oracle construction}\label{sec:oracle_construction}

Oracles can be seen as the interface between a quantum procedure and real-world data.
Their purpose is to encode classical data such that a quantum algorithm can process
it efficiently.

\subsection{Oracle interface}\label{sec:oracle_inout_description}

In order to work as a bridge between the classical and the quantum worlds and to be
used by the quantum algorithm, a clear interface for the oracle should be
established. 

We chose to use the interface described in~\cite[Eq. 4.4]{2004-ahokas-graeme-robert-improved-algorithms-for-approximate-quantum-fourier-transforms-and-sparse-hamiltonian-simulations}
with slight modifications improving the arity of the oracle for our specific case
of $1$-sparse matrices.

More precisely, our oracles $O$ implement the following interface
\begin{equation}
  \label{eq:oracle_interface}
  O \ket{x_0}_x \ket{0}_m \ket{0}_v \ket{0}_s = \ket{x_0}_x \ket{m(x_0)}_m \ket{v(x_0)}_v \ket{s(x_0)}_s
\end{equation}
with $\ket{x_0}_x$ encoding a row index as a unsigned integer, $m(x)$ the function that
returns the column index of the only non-zero element in row $x$,
$v(x) = \vert w(x) \vert$ the absolute value of the weight $w(x)$ of the first (and only)
non-zero element in row $x$ and
\begin{equation}
  \label{eq:oracle_sign_function}
  s(x) = \left\{
    \begin{split}
      0 & \text{ if } w(x) \geqslant 0 \\
      1 & \text{ else }
    \end{split}
  \right.
\end{equation}
the sign of the first non-zero entry in row $x$.

Note that the quantum register are labeled with their respective usage: $x$ for the index of of the row considered, $m$ for the index of the column considered, $v$ for the value of the element at $(\mathrm{row}, \mathrm{index})$ and $s$ for the sign of the element at $(\mathrm{row}, \mathrm{index})$. A fifth label ``$a$'' is used along the paper to label a register used as an ancilla. 

The interface of the oracle $O$ can also be obtained with 3 separate oracles
that will each take care of computing one output:
\begin{equation}
  \label{eq:M-oracle}
  M \ket{x_0}_x \ket{0}_m = \ket{x_0}_x \ket{m(x)}_m
\end{equation}
\begin{equation}
  \label{eq:V-oracle}
  V \ket{x_0}_x \ket{0}_v = \ket{x_0}_x \ket{v(x)}_v
\end{equation}
\begin{equation}
  \label{eq:S-oracle}
  S \ket{x_0}_x \ket{0}_s = \ket{x_0}_x \ket{s(x)}_s
\end{equation}

\subsection{Optimisation of $M$ and $S$}\label{sec:optimization-m-s}

\begin{claim}\label{claim:oracles-ignore-when-V-is-zero}
  The simulation algorithms provided by~\cite{2004-ahokas-graeme-robert-improved-algorithms-for-approximate-quantum-fourier-transforms-and-sparse-hamiltonian-simulations}
  have the interesting property that if the oracle $V$ encodes a weight of
  zero for some inputs (i.e.\ $v(x) = 0$ for some $x$) then the outputs of
  oracles $M$ and $S$ are ignored for those inputs.
\end{claim}

\begin{proof}
  
  The circuit simulating a $1$-sparse $m$-bit-integer weighted Hamiltonian $H$
  depicted in \cref{fig:sparse-ham-ahokas-circuit} is taken from~\cite{2004-ahokas-graeme-robert-improved-algorithms-for-approximate-quantum-fourier-transforms-and-sparse-hamiltonian-simulations}.

  \begin{figure}
    \centering
    \includegraphics[width=.8\linewidth]{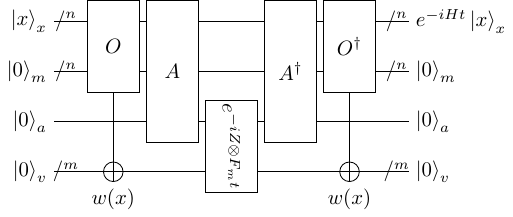}
    \caption{Quantum circuit re-created from~\cite[p.~71]{2004-ahokas-graeme-robert-improved-algorithms-for-approximate-quantum-fourier-transforms-and-sparse-hamiltonian-simulations}
      that simulates a 1-sparse integer-weighted Hamiltonian for a given time $t$.
      $O$ is the implementation of the oracle, $A$ is a quantum circuit defined in~\cite[p.~70]{2004-ahokas-graeme-robert-improved-algorithms-for-approximate-quantum-fourier-transforms-and-sparse-hamiltonian-simulations}.
      $F_m$ is defined as the diagonal matrix with diagonal entries increasing from $0$ to $2^m-1$ (see \cref{eq:matrix-F-def}).
    }\label{fig:sparse-ham-ahokas-circuit}
  \end{figure}
  
  In our special case of $1$-bit weights (i.e.\ $m = 1$), the third quantum gate
  $e^{-i Z \otimes F_m t}$ can be written as
  \begin{equation}
    \label{eq:matrix-form-exp-ZFt}
    \begin{split}
      e^{-i Z \otimes F_m t}
      &= e^{-i Z \otimes F_1 t} \\
      &= \exp\left[ -i 
        \begin{pmatrix}
          F_1 & 0 \\
          0 & -F_1\\
        \end{pmatrix}
        t \right] \\
      &= \begin{pmatrix}
        e^{-iF_1t} & 0 \\
        0 & e^{iF_1t}\\
      \end{pmatrix} \\
      &=
      \begin{pmatrix}
        1 & 0 & 0 & 0 \\
        0 & e^{-it} & 0 & 0 \\
        0 & 0 & 1 & 0 \\
        0 & 0 & 0 & e^{-it} \\
      \end{pmatrix}.
    \end{split}
  \end{equation}

  where 

  \begin{equation}
    \label{eq:matrix-F-def}
    F_m =
    \begin{pmatrix}
      0      & 0      & \cdots & \cdots & \cdots & \cdots & 0      \\
      0      & 1      & \ddots &        &        &        & \vdots \\
      \vdots & \ddots & 2      & \ddots &        &        & \vdots \\
      \vdots &        & \ddots & 3      & \ddots &        & \vdots \\
      \vdots &        &        & \ddots & 4      & \ddots & \vdots \\
      \vdots &        &        &        & \ddots & \ddots & 0      \\
      0      & \cdots & \cdots & \cdots & \cdots & 0      & 2^m-1  \\
    \end{pmatrix}.
  \end{equation}
  
  It follows from the matrix notation that if the second qubit $e^{-iZ\otimes F_1 t}$ is
  applied on is in the state \ket{0}, the gate $e^{-iZ\otimes F_1 t}$ is the identity
  transformation, i.e.\ the unitary operation $e^{-iZ\otimes F_1 t}$ sends \ket{00} (resp. \ket{10})
  to \ket{00} (resp. \ket{10}).
  This means that if the oracle $O$ does not set the last qubit to \ket{1} (i.e.\ the oracle encodes a weight of
  $0$ for the $x$\textsuperscript{th} row of $H$), the quantum circuit
  depicted in \cref{fig:sparse-ham-ahokas-circuit} can be simplified up to an identity
  transformation as the effects of $O$ (resp. $A$) are reverted by $O^\dagger$ (resp. $A^\dagger$).

  Rephrasing, if the $x$\textsuperscript{th} row of matrix $H$ has no non-zero entries, the effects of the oracle $O$ is ignored, which implies that the effects of the oracles $M$ and $S$ that compose $O$ are also ignored.
\end{proof}

Using the result of Claim~\cref{claim:oracles-ignore-when-V-is-zero}, we are free to
implement any transformation that best suits us for the set of inputs $\ket{x}$ such that
the $x$\textsuperscript{th} row of the considered Hermitian matrix ($H_1$ or $H_{-1}$) has no non-zero
elements as long as the oracle $V$ implements the right transformation.

To illustrate clearly the \emph{implemented} transformations we chose to encode with
$M$ and $S$, the next sections will re-write the matrices $H_1$ and $H_{-1}$ according
to \cref{eq:Hi_definition} but with one $\mbf{0}$ or $\mbf{-0}$ in each empty
row. A $\mbf{0}$ entry at position $(i,j)$ in the matrix means that the row $i$ was
empty, the oracle $M$ will map $\ket{i}_x$ to $\ket{j}_m$ and the oracle $S$ will encode a
positive sign, i.e.\ $\ket{0}_s$. The same reasoning applies for $\mbf{-0}$ entries, except
that the encoded sign is now negative, i.e.\ $\ket{1}_s$. 

The following sections will explain step by step the construction of each
of the three oracles $M$, $V$ and $S$, both for the matrix $H_1$ ($M_1$, $V_1$ and $S_1$)
and the matrix $H_{-1}$ ($M_{-1}$, $V_{-1}$ and $S_{-1}$).

\subsection{About arithmetic and logic quantum gates}\label{sec:about-arithm-logic}

Implementing the oracles $M$, $V$ and $S$ for the matrices $H_1$ and $H_{-1}$ requires several
arithmetic and logic quantum gates such as \texttt{or}, \texttt{add} or \texttt{compare}.

All these gates have been implemented prior to the oracle implementation and the implementation
steps are detailed in this section.

\subsubsection{\label{sec:or-gate}The \texttt{or} gate}

The \texttt{or} gate is easily implemented using only \texttt{X} and \texttt{CCX} (or Toffoli) gates.
The implementation used is depicted in \cref{fig:or-gate-implementation} and uses
the famous Boole algebra formula linking \texttt{not}, \texttt{or} and \texttt{and}:
$x \vee y = \neg(\neg x \wedge \neg y)$.

\begin{figure}
  \centering
  \resizebox{\linewidth}{!}{
    \circEqual{%
      \includegraphics{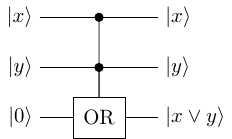}
    }{%
      \includegraphics{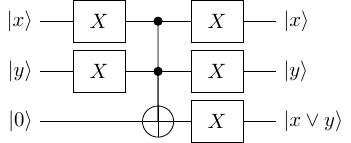}
    }
  }
  \caption{Implementation of the \texttt{or} gate.}\label{fig:or-gate-implementation}
\end{figure}

\subsubsection{\label{sec:add-sub-gates}The \texttt{add} and \texttt{sub} gates}

Most of the research papers presenting an implementation of the \texttt{add} or \texttt{sub}
gates only consider the case where the two numbers to add or subtract are stored in
quantum registers.

In our case, the oracles implementation requires an adder and subtractor that can add or
subtract to a quantum register a quantity known when the quantum circuit is generated, i.e.\ not necessarily encoded on a quantum state.

\begin{claim}\label{claim:subtractor-from-adder}
  Implementing a subtractor is trivial once an adder procedure is available.
\end{claim}
\begin{proof}
  A subtractor can be implemented from a generic adder by using the identity
  \begin{equation}
    \label{eq:bitwise-sub-to-add}
    a - b = (a' + b)'
  \end{equation}
  where $'$ denotes the bitwise complementation.

  The circuit resulting of the application of this identity is depicted in
  \cref{fig:sub-gate-implementation-from-add} and only requires one call
  to the adder and $2n$ additional gates, $n$ being the number of qubits used to
  represent one of the operands.

  \begin{figure}
    \centering
    \resizebox{\linewidth}{!}{
      \circEqual{%
        \includegraphics{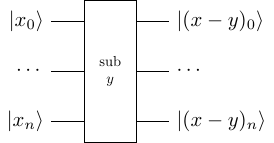}
      }{%
        \includegraphics{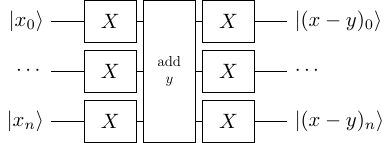}
      }
    }
    \caption{Implementation of the \texttt{sub} gate from an \texttt{add} gate. The $y$ value encoding is intentionally omitted. The subtractor will use the same encoding as the adder (i.e.\ either the $y$ value is encoded on a quantum register or it is encoded directly in the quantum circuit implementing the adder). Note that the $y$ value is not negated.}\label{fig:sub-gate-implementation-from-add}
  \end{figure}

\end{proof}

\begin{note}
  Following Claim~\cref{claim:subtractor-from-adder} we will restrict the study to
  implementing an adder. Implementing a subtractor is trivial and cheap in term
  of additional quantum gates used once an adder is available.
\end{note}

\begin{definition}{\textit{Generation-time value}}
  A generation-time value is a value that is known by the programmer when generating the
  quantum circuit. Knowing a value at generation-time may allow to optimise even further
  the generated quantum circuit. The closest analog in classical programming would be
  C-like macros or recent C++ constexpr expressions.
\end{definition}

The easiest solution to overcome the problem caused by the non-compatible input formats between
our problem (with a generation-time value) and the existing adders (with two values encoded on quantum registers)
is to encode the quantity known at generation-time into ancillary qubits and then use the
regular adder algorithms to add to a quantum register the value encoded in a second quantum register.
Even if this solution is trivial to implement, it has the huge downside of requiring \bigO{\log_2 b}
additional ancillary qubits to temporarily store the generation-time value $b$.

Another answer to the problem would be to adapt a quantum adder originally devised
to add two quantum registers to a quantum adder capable of adding a constant value
to a quantum register. Several adders~\cite{1712.02630v1,quant-ph0410184v1,quant-ph0008033v1}
have been studied to check if they can be modified to allow a generation-time input,
i.e.\ if it possible to remove completely the quantum register storing the right-hand-side
(or left-hand-side) of the addition.


The task of removing the quantum register storing one of the operands
appears to be challenging for adders based on classical arithmetic like~\cite{quant-ph0410184v1,1712.02630v1} but trivial for Draper's quantum adder
introduced in~\cite{quant-ph0008033v1}.


\begin{claim}
  Draper's quantum adder can be adapted into an efficient adder that takes as right-hand side input
  a unsigned ``generation-time'' integer value and add this value to a sufficiently large quantum
  register encoding another unsigned integer.
\end{claim}
\begin{proof}

  The original Drapper's adder as introduced in~\cite{quant-ph0008033v1} is illustrated in
  \cref{fig:original-drapper-adder}.
  
  \begin{figure}
    \centering
    \includegraphics[width=.8\linewidth]{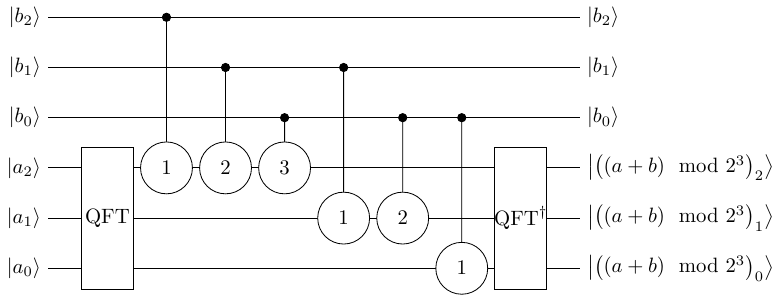}
    \caption{Original Draper's adder example for $3$-qubit registers \ket{a} and \ket{b}. The round gates between the two applications of the Quantum Fourier Transform (\texttt{QFT} gates) are controlled phase gates and are defined in~\cite{quant-ph0008033v1}. Note that the adder wraps on overflow, meaning that if an overflow happens, the result will be $(a+b) \mod 2^3$.
    }\label{fig:original-drapper-adder}
  \end{figure}

  The only quantum gates using the quantum register \ket{b} are the controlled-phase gates.
  Moreover, they only use the qubits of the right-hand-side register \ket{b} as controls.

  In the case of a constant value of $b$ known at generation time, we can replace each
  controlled-phase gate by either a phase gate if the corresponding bit of $b$ is $1$ or
  by an identity gate (or a ``no-op'' gate) if the bit of $b$ is $0$. Once this transformation
  has been performed, the quantum register \ket{b} is no longer used and can be safely removed
  from the circuit. 
\end{proof}

The final quantum \texttt{add} gate implementation is depicted in \cref{fig:modified-drapper-adder},
requires \bigO{n^2} gates and has a depth of \bigO{n}. Following~\cite{quant-ph9601018v1,quant-ph0008033v1,quant-ph0006004v1},
the asymptotic gate count can be improved to \bigO{n\log(n)} by removing the rotation with an angle below
a given threshold that depend on hardware noise.

\begin{figure}
  \centering
  \includegraphics[width=.8\linewidth]{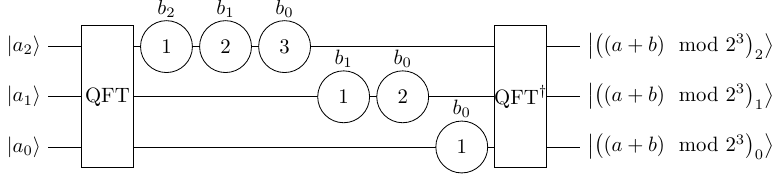}
  \caption{Modified Draper's adder example for $3$-qubit register \ket{a} and $3$-bit classical constant $b$. The round gates between the two applications of the Quantum Fourier Transform (\texttt{QFT} gates) are phase gates and are defined in~\cite{quant-ph0008033v1}. A label $b_i$ above a phase gate means that the phase gate should only be applied when the $i$\textsuperscript{th} bit of $b$ is set to $1$. Note that the adder wraps on overflow, meaning that if an overflow happens, the result will
    be $(a+b) \mod 2^3$.
  }\label{fig:modified-drapper-adder}
\end{figure}


\subsubsection{\label{sec:cmp-gate}The \texttt{cmp} gate}

For the same reasons exposed in the adder implementation in \cref{sec:add-sub-gates},
the \texttt{cmp} gate cannot be implemented using the arithmetic comparator presented
in~\cite{quant-ph0410184v1} because removing the right-hand side qubits seems to be
a challenging task.

Instead, we use the idea from~\cite[Section 4.3]{quant-ph0410184v1} that explain how to
implement a comparator only by using a quantum adder. The comparison algorithm works by
computing the high-bit of the expression $a - b$. If this high-bit is in the state \ket{1}
then $a < b$.

In order to compute the high-bit of $a - b$, several options are open. The two most promising options are described in the following paragraphs.

The first option is to use a subtractor acting on $n+1$ qubits and behaving nicely on underflows (i.e.\ underflows result in cycling to the highest-value), as illustrated in
\cref{fig:high-bit-compute-subtractor}. 
This approach requires $2$ calls to the subtractor and $1$ additional $2$-qubit quantum
gate.

\begin{figure}
  \centering
  \includegraphics[width=.7\linewidth]{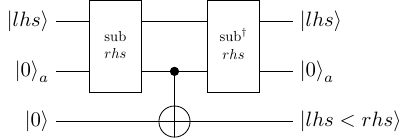}
  \caption{
    Computation of the high-bit of $lhs - rhs$ with a $(n+1)$-qubit subtractor. The second quantum register
    is an ancilla qubit that is appended to the quantum register storing \ket{\mathrm{lhs}} in order to form
    a $(n+1)$-qubit register. The result is stored in a third quantum register as \ket{1} if $lhs < rhs$, else
    \ket{0}.
  }\label{fig:high-bit-compute-subtractor}
\end{figure}

Another solution would be to use \cref{eq:bitwise-sub-to-add} to change the subtraction into
an addition and then use a specialised procedure to compute the high-bit of the addition of
two numbers $a$ and $b$ ($a$ being encoded on a quantum register and $b$ a constant).
Computing the high-bit of an addition between a quantum register and a constant can be performed
with the \texttt{CARRY} gate introduced in~\cite{1611.07995v2}. This approach requires \bigO{n} Toffoli, \texttt{CNOT} and \texttt{X} gates.




Each of the described methods has its advantages and drawbacks.

For example, the first method crucially relies on a quantum subtractor, and will have the same properties as the subtractor used. In our specific case, we use the subtractor implemented with Drapper's adder~\cite{quant-ph0008033v1} as explained in \cref{sec:add-sub-gates}, which in turn uses the quantum Fourier transform. The main disadvantage of using the QFT when looking at practical implementation on quantum hardware is that the QFT involves phase gates with exponentially small angles. These gates may be implemented correctly up to a given threshold, but very small rotation angles will inevitably not be as precise as \textit{normal} rotation angles due to the hardware limitations in precision. This problem can be circumvented by using an approximate QFT algorithm~\cite{quant-ph9601018v1,quant-ph0006004v1} that will cut all the rotation gates that have a rotation angle smaller than a given threshold from the generated circuit but the algorithm will not be exact anymore (small probability of incorrect result).

On the other hand, the \texttt{CARRY} gate involves only \texttt{X},
controlled-\texttt{X} and Toffoli gates. This restriction makes this implementation
more robust than the first one to hardware approximations. Another difference is the
connectivity needed by the approaches: the first method relies on a adder implemented with the
quantum Fourier transform, which use an all-to-all connectity whereas the \texttt{CARRY} gate, once
the qubits correctly ordered, only contains gates on adjacent qubits. As a side note, the exclusive
use of logical gates \texttt{X}, controlled-\texttt{X} and Toffoli may allow us to
simulate efficiently the \texttt{CARRY} gate on classical hardware as it only involves classical
arithmetic.

As a last word, in the future, the QFT may be implemented directly into the
hardware chips to make it more efficient because it is one of the most used
quantum procedure (and so one of the best candidate for optimisation). Taking
this possibility into account seems a little premature right now but may have a high
impact on the efficiency and precision of the first solution presented.

After summarising all the drawbacks and advantages, we decided to use the
arithmetic comparator for its linear number of gates, because it is based on
arithmetic which does not involve exponentially small rotation angles and because
the need to have $n-1$ dirty qubits to lend to the procedure is not an issue in
our implementation.

\subsubsection{\label{sec:eq-gate}The \texttt{eq} gate}

The last gate the oracle implementation will need is an \texttt{eq} gate, testing the
equality between an integer stored in a quantum register and a generation-time constant
integer.

This gate has been implemented with a multi-controlled Tofolli gate and a few \texttt{X} gates
before and after the control qubits of the Toffoli gates that should be equal to \ket{0}.
The \texttt{X} gates are necessary because a raw Toffoli gate set its target qubit only when all
its control are in the state \ket{1}, but we want each control qubit to be equal to a specific
bit of the generation-time constant integer, which can be either \ket{0} or \ket{1}.

An implementation example is available in \cref{fig:eq-implementation}.

\begin{figure}
  \centering
  \resizebox{.7\linewidth}{!}{
    \circEqual{%
      \includegraphics{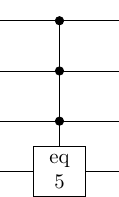}
    }{%
      \includegraphics{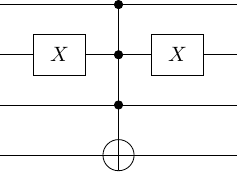}
    }
  }
  \caption{Example of \texttt{eq} gate implementation for the compile-time value $5$. \texttt{X} gates are applied to the second
    control qubit because the only bit set to $0$ in the big-endian binary representation
    of $5 = 101_2$ is at the second (middle) position.}\label{fig:eq-implementation}
\end{figure}

Implementing a \texttt{NOT} gate controlled by $n$ qubits can be done with only one ancilla qubit
or $n-2$ garbage qubits and requires \bigO{n} \texttt{X}, controlled-\texttt{X} or Toffoli gates~\cite{multi.controll.toffoli}.

\subsection{Oracles for $H_1$}\label{sec:oracles-h_1}

As noted in \cref{sec:optimization-m-s}, the oracles $M_1$ and $S_1$ can be
optimized by using the fact that they can encode anything for $\ket{x}_x$ when the
$x$\textsuperscript{th} row of $H_1$ is empty.

We decided to use this optimization opportunity to add regularity to the description
of the $H_1$ matrix. The \emph{implemented} matrix $H_1$, denoted as $H_1^{\text{impl}}$, is described
in \cref{eq:h1i-definition}.

\begin{note}
  All indices start at $0$. The first row of a matrix has the index $0$, the second
  row the index $1$ and so on.
  This convention is used to match Python's indexing that starts at $0$.
\end{note}

\newcommand{\Bmat}{%
  \begin{matrix}%
    1      & 1      & 0      & \cdots & 0      \\%
    0      & -1     & 1      & \ddots & \vdots \\%
    \vdots & \ddots & \ddots & \ddots & \vdots \\%
    \cdots & \cdots & \cdots & \cdots & \vdots \\%
  \end{matrix}%
}
\newcommand{\BTmatrix}{%
  \begin{matrix}%
    \vdots & 0      & \cdots & 0      \\%
    \vdots & -1     & \ddots & \vdots \\%
    \vdots & 1      & \ddots & 0      \\%
    \vdots & \ddots & \ddots & -1     \\%
    0      & \cdots & \cdots & \cdots \\%
  \end{matrix}%
}%
\newcommand{\RightTopZeroMatrix}{%
  \begin{matrix}%
    0      & \cdots & \cdots & \cdots & \cdots & 0      \\%
    \vdots &        &        &        &        & \vdots \\%
    \vdots &        &        &        &        & \vdots \\%
    \vdots &        &        &        &        & \vdots \\%
  \end{matrix}%
}%
\newcommand{\BotLeftZeroMatrix}{%
  \begin{matrix}%
    0      & \cdots & \cdots & \cdots \\%
    \vdots &        &        &        \\%
    \vdots &        &        &        \\%
    \vdots &        &        &        \\%
    \vdots &        &        &        \\%
    0      & \cdots & \cdots & \cdots \\%
  \end{matrix}%
}%

\begin{equation}
  \begin{split}
    H_1^{\text{impl}} =
    \begin{array}{r}
      {
      \setlength{\tabcolsep}{0pt} 
      \begin{array}{rrr}
        \overbrace{\hphantom{\BTmatrix}\hspace{.5\tabcolsep}}^{N_c} & \overbrace{\hphantom{\Bmat}\hspace{.5\tabcolsep}}^{N_c+1} & {\overbrace{\hphantom{\RightTopZeroMatrix}\hspace{.5\tabcolsep}}^{2^q-\left(2N_c + 1\right)}\hspace{3pt}}\\
      \end{array}
      }
      \\[0pt]
      \begin{array}{r}
        \rotatebox[origin=c]{90}{$N_c$}                          \left\{\vphantom{\Bmat}\right.  \\[5pt]
        \rotatebox[origin=c]{90}{$N_c + 1$}                      \left\{\vphantom{\BTmatrix}\right. \\[5pt]
        \rotatebox[origin=c]{90}{$2^q - \left(2N_c+1\right)$}    \left\{\vphantom{\BotLeftZeroMatrix}\right.\\
      \end{array}
      \begin{pmatrix}
        0      & \cdots & \cdots & 0      & 0      & 1      & 0      & \cdots & 0      & 0      & \cdots & \cdots & \cdots & \cdots & 0      \\%
        \vdots &        &        & \vdots & \vdots & \ddots & \ddots & \ddots & \vdots & \vdots &        &        &        &        & \vdots \\%
        \vdots &        &        & \vdots & \vdots &        & \ddots & \ddots & 0      & \vdots &        &        &        &        & \vdots \\%
        0      & \cdots & \cdots & 0      & 0      & \cdots & \cdots & 0      & 1      & 0      &        &        &        &        & \vdots \\%
        0      & \cdots & \cdots & 0      & 0      & \cdots & \cdots & \cdots & 0      & \mbf{0}&        &        &        &        & \vdots \\%
        1      & \ddots &        & \vdots & \vdots &        &        &        & \vdots & 0      &        &        &        &        & \vdots \\%
        0      & \ddots & \ddots & \vdots & \vdots &        &        &        & \vdots & \vdots &        &        &        &        & \vdots \\%
        \vdots & \ddots & \ddots & 0      & \vdots &        &        &        & \vdots & \vdots &        &        &        &        & \vdots \\%
        0      & \cdots & 0      & 1      & 0      & \cdots & \cdots & \cdots & 0      & 0      & \cdots & \cdots & \cdots & \cdots & 0      \\%
        0      & \cdots & \cdots & 0      & \mbf{0}& 0      & \cdots & \cdots & 0      & 0      & \cdots & \cdots & \cdots & \cdots & 0      \\%
        \vdots &        &        &        & \ddots & \mbf{0}& \ddots &        & \vdots & \vdots &        &        &        &        & \vdots \\%
        \vdots &        &        &        &        & \ddots & \mddots& \ddots & \vdots & \vdots &        &        &        &        & \vdots \\%
        \vdots &        &        &        &        &        & \ddots & \mbf{0}& 0      & 0      &        &        &        &        & \vdots \\%
        \vdots &        &        &        &        &        &        & \ddots & \mbf{0}& 0      & \ddots &        &        &        & \vdots \\%
        0      & \cdots & \cdots & \cdots & \cdots & \cdots & \cdots & \cdots & 0      & \mbf{0}& 0      & 0      & \cdots & \cdots & 0      \\%
      \end{pmatrix}
    \end{array}
  \end{split}\label{eq:h1i-definition}
\end{equation}

According to the shape of the matrix in \cref{eq:h1i-definition}, the oracle $M_1$
should implement the transformation
\begin{equation}
  \label{eq:M1-transformation}
  M_{1}\vert x \rangle_x \vert 0 \rangle_m \mapsto
  \begin{cases}
    \vert x \rangle_x \otimes \vert x + (N_c + 1) \rangle_m & \text{if } x < (N_c + 1) \\
    \vert x \rangle_x \otimes \vert x - (N_c + 1) \rangle_m & \text{else} \\
  \end{cases}.
\end{equation}

$M_1$ can be easily implemented with the quantum circuit depicted in 
\cref{fig:M1-implementation-quantum-circuit}.

\begin{figure}
  \centering
  \includegraphics[width=\linewidth]{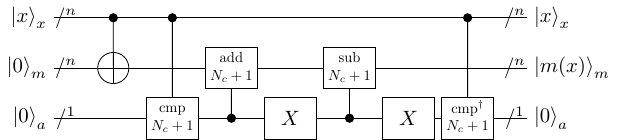}
  \caption{Implementation of the oracle $M_1$. The \texttt{cmp} gate compare
    the value of the control quantum register (interpreted as a unsigned integer)
    with the parameter given (written below the \texttt{cmp}). If the control
    register is strictly lower than the parameter, the gate set the qubit it is applied
    on to \ket{1}. The \texttt{add} (resp.\ \texttt{sub}) gate used in this quantum circuit
    add (resp.\ subtract) the value of its parameter to (resp.\ from) the quantum register
    it is applied on only if the control qubit is in the state \ket{1}.}\label{fig:M1-implementation-quantum-circuit}
\end{figure}

The oracle $V$ cannot be simplified using the results from
Claim~\cref{claim:oracles-ignore-when-V-is-zero}. It should implement
the transformation written in \cref{eq:V1-transformation}.

\begin{equation}
  \label{eq:V1-transformation}
  V_{1}\vert x \rangle_x \vert 0 \rangle_v \mapsto
  \begin{cases}
    \vert x \rangle_x \vert 1 \rangle_v & \text{if } (x < 2N_c + 1) \land (x \neq N_c) \\
    \vert x \rangle_x \vert 0 \rangle_v & \text{else} \\
  \end{cases}.
\end{equation}

The implementation of the oracle $V_1$ is depicted in \cref{fig:V1-implementation-quantum-circuit}. The first part, \textit{Set}, sets the weight qubit to $1$ for all $\ket{x}_x$ such that $x < 2N_c+1$. As this does not correspond to the correct expression of $V$, the second part \textit{Correct} is here to set the weight register back to $\ket{0}_v$ when $x == N_c$. 

\begin{figure}
  \centering
  \includegraphics[width=\linewidth]{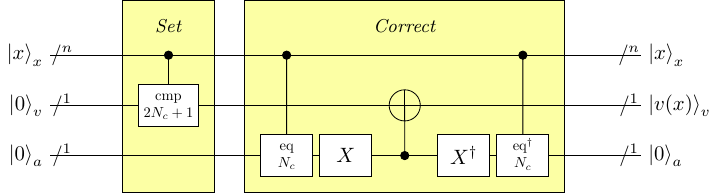}
  \caption{Implementation of the oracle $V_1$. The \texttt{cmp} gate compare
    the value of the control quantum register (interpreted as a unsigned integer)
    with the parameter given (written below the \texttt{cmp}). If the control
    register is strictly lower than the parameter, the gate set the qubit it is applied
    on to \ket{1}. The \texttt{eq} gate used in this quantum circuit sets its target qubit to \ket{1}
    if the value of its parameter is equal to the value encoded on the quantum register controling
    the gate.
  }\label{fig:V1-implementation-quantum-circuit}
\end{figure}

The last oracle left to implement in order to be able to simulate $H_1$ is $S_1$, the
oracle encoding the signs of the non-zero entries of $H_1$.
The convention used to encode the sign of an entry has been taken from~\cite{2004-ahokas-graeme-robert-improved-algorithms-for-approximate-quantum-fourier-transforms-and-sparse-hamiltonian-simulations}
and is: a positive sign is encoded as $\ket{0}_s$, a negative sign is encoded as $\ket{1}_s$. 
As shown in \cref{eq:h1i-definition}, $H_1$ only contains positive non-zero
entries so the sign oracle $S_1$ should implement the simple transformation of
\cref{eq:S1-transformation}: the identity.

\begin{equation}
  \label{eq:S1-transformation}
  S_{1}\vert x \rangle_x \vert 0 \rangle_s \mapsto \vert x \rangle_x \vert 0 \rangle_s
\end{equation}

\subsection{Oracles for $H_{-1}$}\label{sec:oracles-h_-1}

The matrix $H_{-1}$ has less regularity than $H_1$, which will lead to a more complex
implementation. The \emph{implemented} matrix $H_{-1}$, denoted as $H_{-1}^{\text{impl}}$,
is described in \cref{eq:h-1i-definition}.

\newcommand{\BotRightZeroMatrix}{%
  \begin{matrix}%
    0      & \cdots & \cdots & \cdots & \cdots & 0      \\%
    \vdots &        &        &        &        & \vdots \\%
    \vdots &        &        &        &        & \vdots \\%
    0      &        &        &        &        & \vdots \\%
    \mbf{0}& \ddots &        &        &        & \vdots \\%
    0      &-\mbf{0}& 0      & \cdots & \cdots & 0      \\%
    
  \end{matrix}%
}%
\newcommand{\BotMidZeroMatrix}{%
  \begin{matrix}%
    0      &-\mbf{0}& \ddots &        & \vdots \\
    & \ddots & \ddots & \ddots & \vdots \\
    &        & \ddots & \ddots & 0      \\
    &        &        & \ddots &-\mbf{0}\\
    &        &        &        & \ddots \\
    \cdots & \cdots & \cdots & \cdots & \cdots \\
  \end{matrix}%
}%

\begin{equation}
  \begin{split}
    H_{-1}^{\text{impl}} =
    \begin{array}{r}
      {
      \setlength{\tabcolsep}{0pt} 
      \begin{array}{rrr}
        \overbrace{\hphantom{\BTmatrix}\hspace{1ex}}^{N_c} & \overbrace{\hphantom{\BotMidZeroMatrix}\hspace{.5ex}}^{N_c+1} & {\overbrace{\hphantom{\BotRightZeroMatrix\hspace{2ex}}}^{2^q-\left(2N_c + 1\right)}\hspace{3pt}}\\
      \end{array}
      }
      \\[0pt]
      \begin{array}{r}
        \rotatebox[origin=c]{90}{$N_c$}                          \left\{\vphantom{\Bmat}\right.  \\[5pt]
        \rotatebox[origin=c]{90}{$N_c + 1$}                      \left\{\vphantom{\BTmatrix}\right. \\[5pt]
        \rotatebox[origin=c]{90}{$2^q - \left(2N_c+1\right)$}    \left\{\vphantom{\BotMidZeroMatrix}\right.\\
      \end{array}
      \begin{pmatrix}%
        0      & \cdots & \cdots & 0      & 1      & 0      & \cdots & \cdots & 0      & 0      & \cdots & \cdots & \cdots & \cdots & 0      \\%
        \vdots &        &        & \vdots & 0      & -1     & \ddots &        & \vdots & \vdots &        &        &        &        & \vdots \\%
        \vdots &        &        & \vdots & \vdots & \ddots & \ddots & \ddots & \vdots & \vdots &        &        &        &        & \vdots \\%
        0      & \cdots & \cdots & 0      & 0      & \cdots & 0      & -1     & 0      & \vdots &        &        &        &        & \vdots \\%
        1      & 0      & \cdots & 0      & 0      & \cdots & \cdots & \cdots & 0      & \vdots &        &        &        &        & \vdots \\%
        0      & -1     & \ddots & \vdots & \vdots &        &        &        & \vdots & \vdots &        &        &        &        & \vdots \\%
        \vdots & \ddots & \ddots & 0      & \vdots &        &        &        & \vdots & \vdots &        &        &        &        & \vdots \\%
        \vdots &        & \ddots & -1     & 0      &        &        &        & \vdots & \vdots &        &        &        &        & \vdots \\%
        0      & \cdots & \cdots & 0      &-\mbf{0}& \ddots &        &        & \vdots & 0      & \cdots & \cdots & \cdots & \cdots & 0      \\%
        0      & \cdots & \cdots & \cdots & 0      &-\mbf{0}& \ddots &        & \vdots & 0      & \cdots & \cdots & \cdots & \cdots & 0      \\%
        \vdots &        &        &        &        & \ddots & \ddots & \ddots & \vdots & \vdots &        &        &        &        & \vdots \\%
        \vdots &        &        &        &        &        & \ddots & \ddots & 0      & \vdots &        &        &        &        & \vdots \\%
        \vdots &        &        &        &        &        &        & \ddots &-\mbf{0}& 0      &        &        &        &        & \vdots \\%
        \vdots &        &        &        &        &        &        &        & \ddots &-\mbf{0}& \ddots &        &        &        & \vdots \\%
        0      & \cdots & \cdots & \cdots & \cdots & \cdots & \cdots & \cdots & \cdots & 0      &-\mbf{0}& 0      & \cdots & \cdots & 0      \\%
      \end{pmatrix}
    \end{array}
  \end{split}\label{eq:h-1i-definition}
\end{equation}

Following the placement of the non-zero and the \mbf{0} or $-\mbf{0}$ entries in the matrix
$H_{-1}^{\text{impl}}$ of \cref{eq:h-1i-definition}, the oracle $M_{-1}$ should implement
the transformation
\begin{equation}
  \label{eq:m-1-transformation}
  M_{-1}\vert x \rangle_x \vert 0 \rangle_m \mapsto
  \begin{cases}
    \vert x \rangle_x \vert x + N_c \rangle_m & \text{if } x < N_c \\
    \vert x \rangle_x \vert x - N_c \rangle_m & \text{else} \\
  \end{cases}.
\end{equation}

This transformation is quite similar to the one implemented by the oracle $M_1$ in
\cref{eq:M1-transformation}: $N_c + 1$ from the transformation $M_1$ has
been replaced by $N_c$ in the transformation $M_{-1}$. Thanks to this similarity,
the implementation of $M_{-1}$ will be a nearly-exact copy of the implementation of
$M_1$. The full implementation of the $M_{-1}$ oracle is depicted in 
\cref{fig:M-1-implementation-quantum-circuit}.

\begin{figure}
  \centering
  \includegraphics[width=\linewidth]{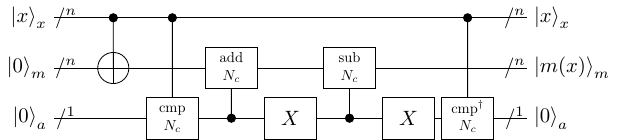}
  \caption{Implementation of the oracle $M_{-1}$. The \texttt{cmp} gate compare
    the value of the control quantum register (interpreted as a unsigned integer)
    with the parameter given (written below the \texttt{cmp}). If the control
    register is strictly lower than the parameter, the gate set the qubit it is applied
    on to \ket{1}. The \texttt{add} (resp. \texttt{sub}) gate used in this quantum circuit
    add (resp.\ subtract) the value of its parameter to (resp.\ from) the quantum register
    it is applied on only if the control qubit is in the state \ket{1}.}\label{fig:M-1-implementation-quantum-circuit}
\end{figure}

The weight oracle $V_{-1}$ is the simplest to implement for the matrix $H_{-1}$, even if
it cannot take advantage of the optimisation discussed in
Claim~\cref{claim:oracles-ignore-when-V-is-zero}.
The transformation that should be implemented by the oracle $V_{-1}$ is shown in
\cref{eq:V-1-transformation}.

\begin{equation}
  \label{eq:V-1-transformation}
  V_{-1}\vert x \rangle_x \vert 0 \rangle_v \mapsto
  \begin{cases}
    \vert x \rangle_x \vert 1 \rangle_v & \text{if } x < 2N_c \\
    \vert x \rangle_x \vert 0 \rangle_v & \text{else} \\
  \end{cases}.
\end{equation}

The implementation of the weight oracle $V_{-1}$ is illustrated in 
\cref{fig:V-1-implementation-quantum-circuit}.

\begin{figure}
  \centering
  \includegraphics[width=.4\linewidth]{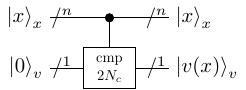}
  \caption{Implementation of the oracle $V_{-1}$. The \texttt{cmp} gate compare
    the value of the control quantum register (interpreted as a unsigned integer)
    with the parameter given (written below the \texttt{cmp}). If the control
    register is strictly lower than the parameter, the gate set the qubit it is applied
    on to \ket{1}.}\label{fig:V-1-implementation-quantum-circuit}
\end{figure}

The last oracle left to implement is $S_{-1}$, the sign oracle. Due to the sign irregularity
in the matrix $H_{-1}^{\text{impl}}$, the implementation of $S_{-1}$ is more involved and requires several ancillary
qubits.
According to the shape of the matrix $H_{-1}^{\text{impl}}$, the sign oracle $S_{-1}$ should implement the
transformation defined in \cref{eq:S-1-transformation}.

\begin{equation}
  \label{eq:S-1-transformation}
  S_{-1}\vert x \rangle_x \vert 0 \rangle_s \mapsto
  \begin{cases}
    \vert x \rangle_x \vert 0 \rangle_s & \text{if } (x = 0) \lor (x = N_c)\\
    \vert x \rangle_x \vert 1 \rangle_s & \text{else} \\
  \end{cases}.
\end{equation}

An implementation of the oracle $S_{-1}$ is illustrated in \cref{fig:S-1-implementation-quantum-circuit}. 
\begin{figure}
  \centering
  \includegraphics[width=.8\linewidth]{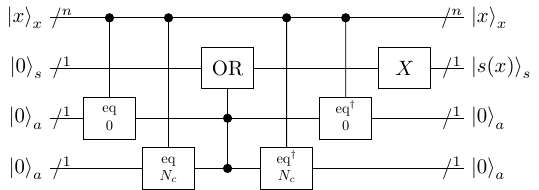}
  \caption{Implementation of the oracle $S_{-1}$. The \texttt{eq} gate used in this quantum circuit is presented in \cref{fig:eq-implementation} and test if the value encoded in its control qubits is equal to the compile-time value given. The \texttt{OR} gate flips the target qubits if and only if at least one of the two control qubits is in the state \ket{1}.}\label{fig:S-1-implementation-quantum-circuit}
\end{figure}

\section{Implementation of higher-order Laplacians}\label{sec:impl-high-order}

The results shown in this paper have all been generated using the second-order discretisation formula shown in \cref{eq:discr_approx}. Higher-order formulas are studied in~\cite[VI, VII and C]{1711.05394v1}.

\begin{note}
  As shown in~\cite[VII.B]{1711.05394v1}, higher-order discretisations with Neumann boundary conditions are not implementable using the algorithm described in \cref{sec:herm_matrix_construction_decomposition}.
\end{note}

In this appendix we replace the second-order formula given in \cref{eq:discr_approx} and used all along the paper by the fourth-order formula given in~\cite[Eq. (46)]{1711.05394v1} and re-written below

\begin{equation}
  \label{eq:fourth-order-laplacian}
  \frac{\partial^2 \phi}{\partial x^2}(i \delta x, t) \approx \frac{1}{{\delta x}^2} \left( \frac{5}{2} \phi_{i,t} - \frac{4}{3} \left(\phi_{i-1,t} + \phi_{i+1,t}\right) + \frac{1}{12} \left( \phi_{i-2,t} + \phi_{i+2,t} \right) \right).
\end{equation}

We are left to devise the matrix $B_d^4$ that satisfy $B_d^4 {B_d^4}^\dagger = \Delta_4$ where $\Delta_4$ is the discretisation matrix arising from the fourth-order finite-differences approximation in~\cref{eq:fourth-order-laplacian}.

\cite[Eq. (47) and VII.C]{1711.05394v1} devised an analytic formula for $\hat{B}_d^4$, the $B_d^4$ matrix with periodic boundary conditions, using the matrix $\hat{S}$ representing the cyclic permutation $\{1, 2, \dots, N\}$ with entries $\hat{S}_{i,j} = \delta_{i, (j+1 \mod N)}$ as shown in \cref{eq:S-permutation-matrix}.

\begin{equation}
  \label{eq:S-permutation-matrix}
  \hat{S} =
  \begin{pmatrix}
    0      & 0      & \cdots & \cdots & \cdots & 0      & 1      \\
    1      & 0      &        &        &        &        & 0      \\
    0      & \ddots & \ddots &        &        &        & \vdots \\
    \vdots & \ddots & \ddots & \ddots &        &        & \vdots \\
    \vdots &        & \ddots & \ddots & \ddots &        & \vdots \\
    \vdots &        &        & \ddots & \ddots & \ddots & \vdots \\
    0      & \cdots & \cdots & \cdots & 0      & 1      & 0      \\
  \end{pmatrix}
\end{equation}

With this definition of $\hat{S}$, the analytic formula for $\hat{B}_d^4$ is given in \cref{eq:Bd4-expression}, with $b$ and $c$ being solution of~\cite[Eqs. (53,54,55)]{1711.05394v1}. The exact values for $b$ and $c$ are:
\begin{align}
  b = \pm \frac{1}{2\sqrt{3}\sqrt{7 \pm 4 \sqrt{3}}}\\
  c = \pm \sqrt{\frac{7 \pm 4\sqrt{3}}{12}}
\end{align}
with the $\pm$ signs that can be chosen freely. Note that, $b = \pm \frac{1}{2\sqrt{3}\sqrt{7 \pm 4\sqrt{3}}}$ and $c = \frac{1}{12b}$ are irrational because $\sqrt{3}\sqrt{7 \pm 4\sqrt{3}} = \sqrt{3}\sqrt{2 + \sqrt{3}}^2 = \sqrt{3}\left(2 + \sqrt{3}\right) = 2\sqrt{3} + 3$ is irrational.

\cref{eq:Bd4-approx-matrix} shows the matrix shape with its entries.
\begin{align}
  \hat{B}_d^4 &= c \hat{S} - (b + c) * \mathbb{I} + b \hat{S}^\dagger \label{eq:Bd4-expression} \\
              &\approx
                \begin{pmatrix}
                  b + c  & b      & 0      & \cdots & \cdots & \cdots & 0      & c      \\
                  c      & b + c  & b      & \ddots &        &        &        & 0      \\
                  0      & c      & \ddots & \ddots & \ddots &        &        & \vdots \\
                  \vdots & \ddots & \ddots & \ddots & \ddots & \ddots &        & \vdots \\
                  \vdots &        & \ddots & \ddots & \ddots & \ddots & \ddots & \vdots \\
                  \vdots &        &        & \ddots & \ddots & \ddots & b      & 0      \\
                  0      &        &        &        & \ddots & c      & b + c  & b      \\
                  b      & 0      & \cdots & \cdots & \cdots & 0      & c      & b + c  \\
                \end{pmatrix} \label{eq:Bd4-approx-matrix}
\end{align}

Because periodicity has not been studied in the main use-case of this paper, we would like to also remove the need of periodic boundary conditions in this higher-order laplacian discretisation. This can be achieved by removing the upper-right entry of $\hat{S}$ by changing it from $1$ to $0$. The resulting matrix $S$ is shown in \cref{eq:S-permutation-matrix-no-periodic}.

\begin{equation}
  \label{eq:S-permutation-matrix-no-periodic}
  S =
  \begin{pmatrix}
    0      & 0      & \cdots & \cdots & \cdots & \cdots & 0      \\
    1      & 0      &        &        &        &        & \vdots \\
    0      & \ddots & \ddots &        &        &        & \vdots \\
    \vdots & \ddots & \ddots & \ddots &        &        & \vdots \\
    \vdots &        & \ddots & \ddots & \ddots &        & \vdots \\
    \vdots &        &        & \ddots & \ddots & \ddots & \vdots \\
    0      & \cdots & \cdots & \cdots & 0      & 1      & 0      \\
  \end{pmatrix}
\end{equation}

Using the exact same formula we can devise $B_d^4$:
\begin{align}
  B_d^4 &= c S - (b + c) * \mathbb{I} + b S^\dagger \label{eq:Bd4-expression-no-periodicity} \\
        &\approx
          \begin{pmatrix}
            b + c  & b      & 0      & \cdots & \cdots & \cdots & \cdots & 0      \\
            c      & b + c  & b      & \ddots &        &        &        & \vdots \\
            0      & c      & \ddots & \ddots & \ddots &        &        & \vdots \\
            \vdots & \ddots & \ddots & \ddots & \ddots & \ddots &        & \vdots \\
            \vdots &        & \ddots & \ddots & \ddots & \ddots & \ddots & \vdots \\
            \vdots &        &        & \ddots & \ddots & \ddots & b      & 0      \\
            \vdots &        &        &        & \ddots & c      & b + c  & b      \\
            0      & \cdots & \cdots & \cdots & \cdots & 0      & c      & b + c  \\
          \end{pmatrix} \label{eq:Bd4-approx-matrix-no-periodicity}
\end{align}

Replacing $B_d^4$ in \cref{eq:H_definition} we obtain
\begin{equation}
  \label{eq:H_matrix_higher_order}
  \widetilde{H_d^4} = \frac{1}{\delta x}
  \begin{pmatrix}
    0      & \cdots & \cdots & \cdots & 0      & b + c  & b      & 0      & \cdots & 0      \\
    \vdots &        &        &        & \vdots & c      & b + c  & b      & \ddots & \vdots \\
    \vdots &        &        &        & \vdots & 0      & c      & \ddots & \ddots & 0      \\
    \vdots &        &        &        & \vdots & \vdots & \ddots & \ddots & b + c  & b      \\
    0      & \cdots & \cdots & \cdots & 0      & 0      & \cdots & 0      & c      & b + c  \\
    b + c  & c      & 0      & \cdots & 0      & 0      & \cdots & \cdots & \cdots & 0      \\
    b      & b + c  & \ddots & \ddots & \vdots & \vdots &        &        &        & \vdots \\
    0      & \ddots & \ddots & \ddots & 0      & \vdots &        &        &        & \vdots \\
    \vdots & \ddots & \ddots & b + c  & c      & \vdots &        &        &        & \vdots \\
    0      & \cdots & 0      & b      & b + c  & 0      & \cdots & \cdots & \cdots & 0      \\
  \end{pmatrix}.
\end{equation}

One of the main difference with the second-order approximation used all along this paper is that, with the fourth-order approximation, the entries of the matrix $\widetilde{H_d^4}$ are no longer multiples of a common number $\alpha \in \mathbb{R}$. This means that we cannot write down $\widetilde{H_d^4}$ as an integer weighted matrix multiplied by a real number, and so the trick used in \cref{sec:matrices_construction} to simulate the integer weighted matrix $H_d$ for a time $\alpha t$ is no longer applicable.

Consequently, and independently of the decomposition we use for $\hat{H_d^4}$, at least one of the matrices in the decomposition of $\hat{H_d^4}$ will not be ``easy to simulate'' as defined in Definition \cref{def:easy-to-simulate}.

Ultimately, the main consequence of this observation is that we will have to use a real-weighted hamiltonian simulation procedure. Such a procedure can be found in~\cite{2004-ahokas-graeme-robert-improved-algorithms-for-approximate-quantum-fourier-transforms-and-sparse-hamiltonian-simulations} but requires to approximate the real-weighted entries with a fixed-point representation that has at least 2 evident caveats:

\begin{enumerate}
\item It is impossible to encode the irrational numbers $b$ and $c$ exactly with a fixed-point representation. This means that we add another layer of approximation, even before the approximation caused by the use of a product-formula.
\item The hamiltonian simulation procedures used for real numbers requires more qubits. More precisely, the number of additional qubits required depends on the desired precision $\epsilon$ and grows as $\log_2\left(\frac{1}{\epsilon}\right)$. 
\end{enumerate}

\begin{note}
  Even if the $H_d^4$ matrix seems quite hard to simulate, it is still a $3$-sparse matrix. This means that it is still managable to hand-write the oracles. Moreover, having a small number of matrices in the decomposition helps in reducing the error introduced by product-formulas.
\end{note}

\section{Optimisation of the implementation}\label{sec:optim-impl}

Once the correctness of the implementation validated, one of the most important remaining work is to try to optimise the implementation. The optimisation of a software is often performed as an iterative task.

The first step is to define a figure of merit, a quantity we want to minimise during the optimisation process. Among the most obvious figures of merit are the total number of gates, the number of CNOT gates or the total execution time of the quantum program. More complex quantities can also be considered, such as the execution time using error correction codes or the final state fidelity. In this paper we decided to take into account an estimation of the total execution time of the quantum program on an imaginary device that shares today's chips characteristics.

The second step of the optimisation process consists in isolating the subroutines that contribute the most to the figure of merit. As an example, if the quantum program spend $90\%$ of the total execution time in one subroutine, this subroutine should be the first place to look for optimisations.

After the isolation of one or two subroutines, the actual optimisation can take place. The goal of this third step is to decrease the impact of the subroutines considered on the overall figure of merit without changing the final result of the implementation.

Finally, once the optimisation is performed, the optimisation process can be repeated by re-starting at the second step, until the program is considered sufficiently optimised.

One of the main difficulty we encountered when applying this optimisation process was to correctly isolate the most time-consumming subroutines. In classical computing, this step is usually performed with specialised tools such as \texttt{gprof} or a more advanced profiler, but no such tool exist for quantum programs. In order to fill this gap we developped \texttt{qprof}, a tool that analyses a quantum program and generates a report similar to the one generated by \texttt{gprof}. Using \texttt{qprof} and some of the various tools compatible with \texttt{gprof}, we plotted the call-graph shown in \cref{fig:qft-based-adder}.

From this call-graph, it is clear that the adder is \emph{the} most costly subroutine and that it should be optimised. The adder internally uses the Quantum Fourier Transform (QFT), which takes more than $50\%$ of the total execution time. The issue is that the QFT implementation is already very concise and we do not expect to be able to optimise it enough to cut significantly its overall cost. This leads us to the conclusion that a new algorithm that do not require the QFT should be used to implement an adder. Such an algorithm can be found in~\cite{Vedral_1996}.

Changing the implementation of the adder from Draper's adder to the arithmetic-based adder from~\cite{Vedral_1996} improves drastically the total execution time of the quantum program and produce the  call-graph in \cref{fig:arith-based-adder}.

\begin{figure}[h]
  \centering
  \subfigure[%
  {Call graph of the quantum wave equation solver using Draper's adder (QFT-based).}%
  ]{
    \includegraphics[width=.4\linewidth]{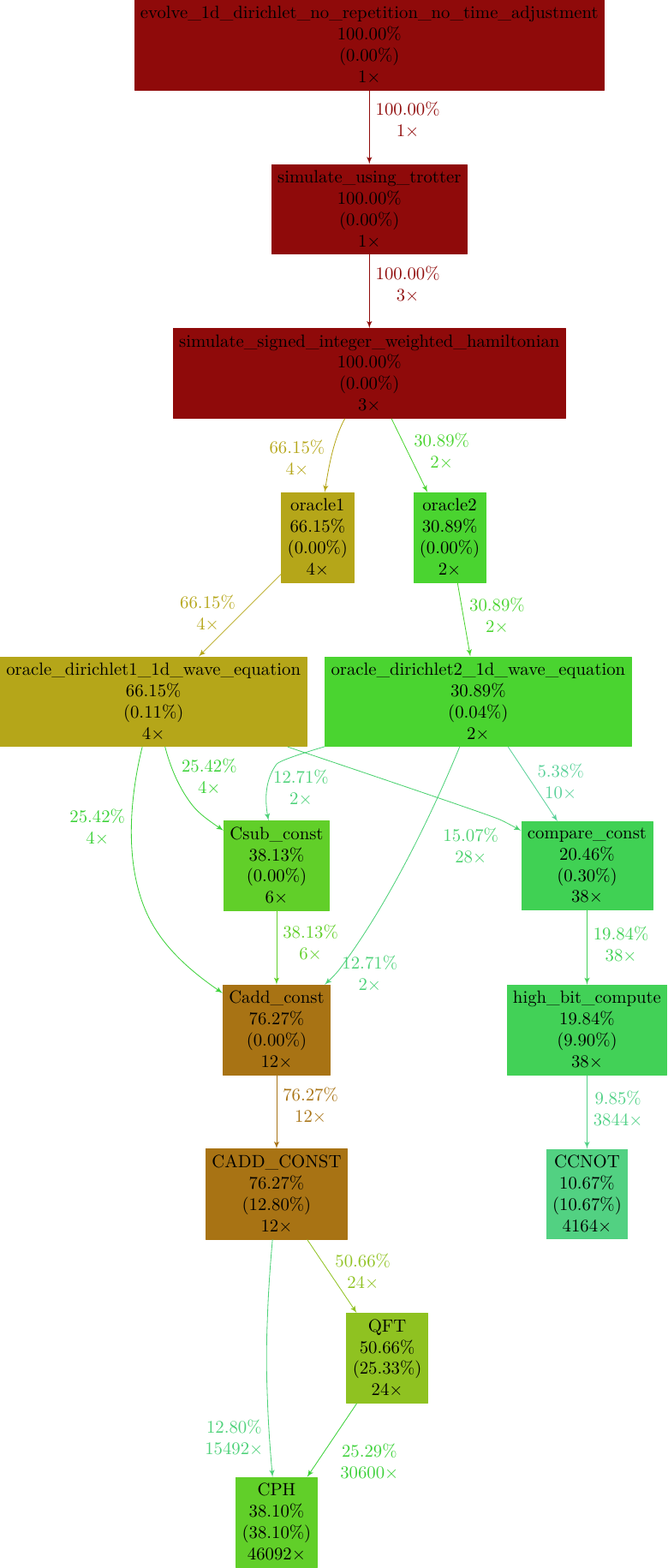}%
    \label{fig:qft-based-adder}
  }\hspace{.03\linewidth}
  \subfigure[%
  {Call graph of the quantum wave equation solver using an arithmetic-based adder.}%
  ]{
    \includegraphics[width=.45\linewidth]{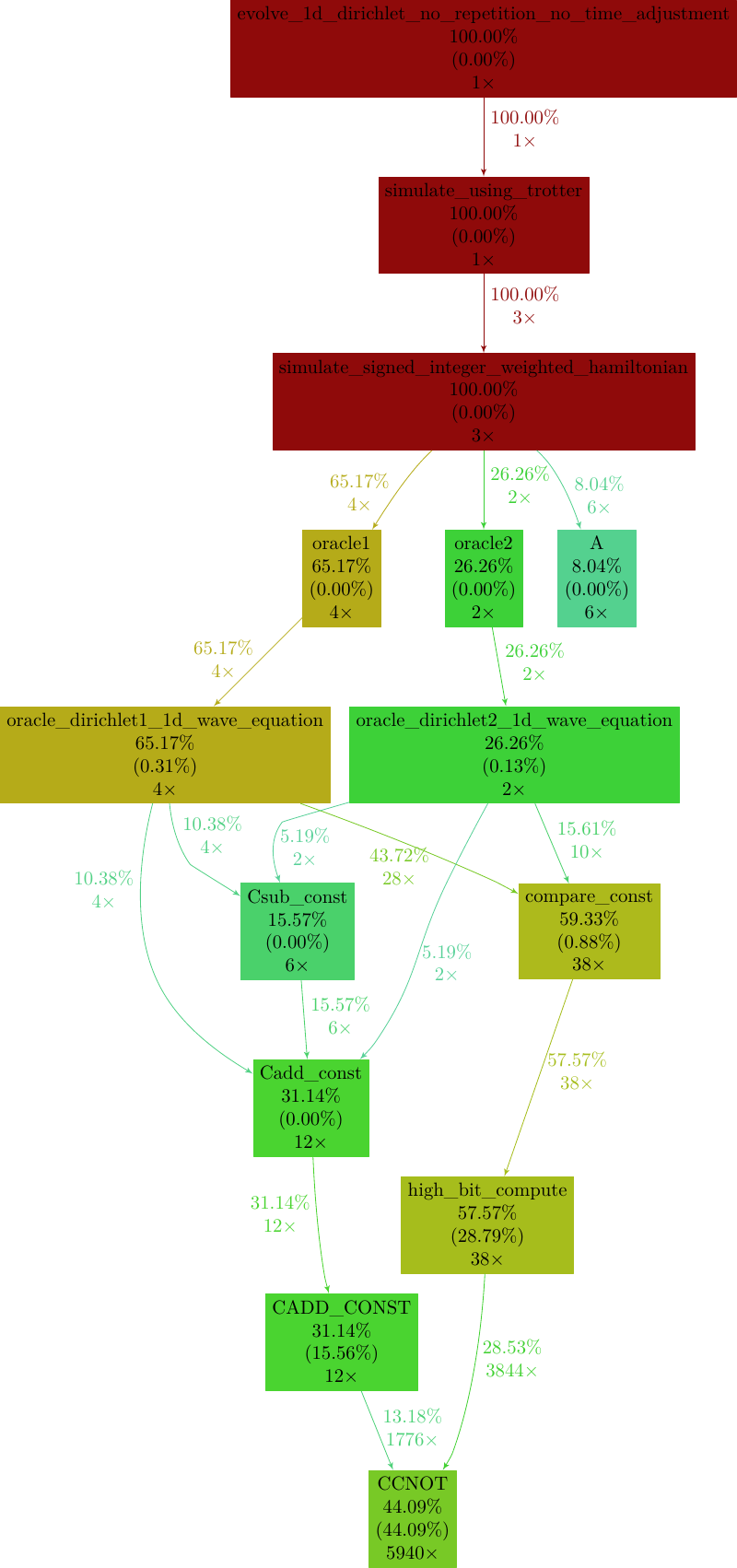}%
    \label{fig:arith-based-adder}
  }
  \caption{All the gates or subroutines that account for less than $5\%$ of the total execution time are not displayed. Execution times for $u1$, $u2$, $u3$ and $cx$ gates have been averaged over all the data available for the quantum chip IBMQ Melbourne. Using the arithmetic-based adder, the overall execution time improved by a factor of $31$.}\label{fig:call-graphs}
\end{figure}

\section{Gate count analysis}
\label{sec:precise-gate-count}

\subsection{Precise subroutines gate counts}
\label{sec:prec-subr-gate}

\cref{tab:gate_counts} summarise the gate count and ancilla qubit requirements for all the major subroutines used in the wave equation solver implementation. Using the entries of this table, it is
possible to compute an estimation of the number of gates required to solve the wave equation. 

As explained in \cref{sec:pf_impl_details}, we need to simulate each of the $1$-sparse Hamiltonians in the decomposition. Aggregating the estimates in \cref{tab:gate_counts} we obtain the costs in \cref{tab:HS_gate_counts} for the Hamiltonian simulation part. Note that the cost of the adder has been voluntarily omited from the computations in order to be able to compare the cost with different adder implementations. Let $a$ be the gate cost of the adder implementation chosen, the cost of simulating the Hamiltonian needed to solve the $1$-dimensional wave equation is: $82n - 35 + 12a_{\mathtt{Toffoli}}$ Toffoli gates, $84n + 21 + 9\left[ 0, n-1 \right] + 12a_{\mathtt{CNOT}}$ \texttt{CNOT} gates and $\bigO{n} + 12a_{1-\mathtt{qubit}}$ $1$-qubit gates.

Choosing an adder implementation and simplifying the gate counts by omitting negligible terms we obtain the gate counts summarised in \cref{tab:final_HS_gate_counts}. It is interesting to note that even if the arithmetic-based adder adds huge constants in the gate count, it does not change the asymptotic complexity whereas Draper's adder changes the number of \texttt{CNOT} gates required from \bigO{n} to \bigO{n^2}.

\begin{table}[h]
  \centering
  \resizebox{\linewidth}{!}{%
    \begin{tabular}{|l|c|c|c|c|m{2.1cm}|}
      \hline
      Gate & Toffoli count & \texttt{CNOT} count & $1$-qubit gate count & \# ancillas & notes\\
      \hline
      \hline
      \texttt{or} & 1 & 0 & 5 & 0 & \\
      \hline
      \texttt{QFT} & 0 & $3\left(2n^2 - 2n + \lfloor \frac{n}{2} \rfloor\right)$&%
                                                                                  \begin{tabular}{@{}c@{}}
                                                                                    $2\left( n^2 + n \right)$ \texttt{$H$} \\
                                                                                    $4\left( n^2 - n \right)$ \texttt{$T$} \\
                                                                                    $\frac{n^2 - n}{2}$ \texttt{$R_n$}
                                                                                  \end{tabular} & 1 \ket{0}-init & $R_n$ gates might need to be decomposed \cite{ross2014optimal}.\\
      \hline
      \texttt{add\_arith} & $20n - 10$ & $22n$ & $0$ & $n-1$ \ket{0}-init & See \cite{Vedral_1996}.\\
      \hline
      \texttt{add\_qft} & 0 & $6\left(2n^2 - 2n + \lfloor \frac{n}{2} \rfloor\right)$&%
                                                                                       \begin{tabular}{@{}c@{}}
                                                                                         $2\left( n^2 + n \right)$ \texttt{$H$} \\
                                                                                         $4\left( n^2 - n \right)$ \texttt{$T$} \\
                                                                                         $3\frac{n^2 - n}{2}$ \texttt{$R_n$}
                                                                                       \end{tabular} & 1 \ket{0}-init & See \texttt{QFT} note on $R_n$. \cref{fig:modified-drapper-adder}.\\
      \hline
      \texttt{sub\_qft} & 0 & $6\left(2n^2 - 2n + \lfloor \frac{n}{2} \rfloor\right)$&%
                                                                                       \begin{tabular}{@{}c@{}}
                                                                                         $2\left( n^2 + n \right)$ \texttt{$H$} \\
                                                                                         $4\left( n^2 - n \right)$ \texttt{$T$} \\
                                                                                         $3\frac{n^2 - n}{2}$ \texttt{$R_n$} \\
                                                                                         $2n$ \texttt{$X$}
                                                                                       \end{tabular} & 1 \ket{0}-init & See \texttt{QFT} note on $R_n$. \cref{fig:sub-gate-implementation-from-add}.\\
      \hline
      \texttt{CARRY} & $2(n - 1)$ & $2 + \left[ 0, n-1 \right]$ & $2n + \left[ 0, n-1 \right]$ \texttt{$X$} & $n-1$ borrowed & See \cite{1611.07995v2}.\\
      \hline
      \texttt{$n$-contr. \texttt{CNOT}} & $4n$ & $0$ & $0$ & $n$ borrowed & See \cite{multi.controll.toffoli}.\\
      \hline
      \texttt{eq} & $4n$ & $0$ & $2 \left[ 0, n \right]$ \texttt{$X$}& $n$ borrowed & \cref{fig:eq-implementation}.\\
      \hline
      \texttt{cmp}& $2\left( n-1 \right)$ & $2 + \left[ 0, n - 1\right]$ & $4n + \left[ 0, n-1 \right]$ \texttt{$X$} & $n-1$ borrowed & See \texttt{CARRY} and \cref{sec:cmp-gate}.\\
      \hline
      \texttt{A} & $2n$ & $4n$ &%
                                 \begin{tabular}{@{}c c@{}}
                                   $3n$ \texttt{$H$} & 
                                                       $3n$ \texttt{$S$} \\
                                   $2n$ \texttt{$T$} & 
                                                       $2n$ \texttt{$X$} \\
                                 \end{tabular} & 0 & See \cite[Fig. 4.3.]{2004-ahokas-graeme-robert-improved-algorithms-for-approximate-quantum-fourier-transforms-and-sparse-hamiltonian-simulations}.\\
      \hline
      \texttt{$e^{-i Z \otimes Z \otimes F t}$} & $8n$ & $24n$ &%
                                                                 \begin{tabular}{@{}c@{}}
                                                                   $36n$ \texttt{$P_h$} \\
                                                                   $8$ \texttt{$X$}
                                                                 \end{tabular}%
           & 0 & Adapted from \cite[Fig. 4.6]{2004-ahokas-graeme-robert-improved-algorithms-for-approximate-quantum-fourier-transforms-and-sparse-hamiltonian-simulations} \\
      \hline
      \texttt{$1$-sparse HS} & $10n$ & $28n$ &%
                                               \begin{tabular}{@{}c c@{}}
                                                 $3n$ \texttt{$H$} & $3n$ \texttt{$S$} \\
                                                 $2n$ \texttt{$T$} & $2n + 8$ \texttt{$X$} \\
                                                 $36n$ \texttt{$P_h$}
                                               \end{tabular} & 0 & Oracle implementation cost not included. $2$ calls to the oracle are required. \cref{fig:sparse-ham-ahokas-circuit}.\\
      \hline
      $M_1$ & $4 \left( n-1 \right)$ & $5 + 2\left[ 0, n-1 \right]$ & $10n + 2 + \left[ 0, n-1 \right]$ \texttt{$X$} &%
                                                                                                                       \begin{tabular}{@{}c@{}}
                                                                                                                         $1$ \ket{0}-init \\
                                                                                                                         $n-1$ borrowed
                                                                                                                       \end{tabular}%
           & \texttt{add} implementation cost not included. $2$ calls to \texttt{add} are required. \cref{fig:M1-implementation-quantum-circuit}.\\
      \hline
      $V_1$ & $2\left( n-1 \right)$ & $2 + \left[ 0, n-1 \right]$ & $4n + \left[ 0, n-1 \right]$ \texttt{$X$} & $n-1$ borrowed & \cref{fig:V1-implementation-quantum-circuit}.\\
      \hline 
      $S_1$ & $0$ & $0$ & $0$ & $0$ & \cref{eq:S1-transformation}.\\
      \hline
      $M_{-1}$ & $4 \left( n-1 \right)$ & $5 + 2\left[ 0, n-1 \right]$ & $10n + 2 + \left[ 0, n-1 \right]$ \texttt{$X$} &%
                                                                                                                          \begin{tabular}{@{}c@{}}
                                                                                                                            $1$ \ket{0}-init \\
                                                                                                                            $n-1$ borrowed
                                                                                                                          \end{tabular}%
           & \texttt{add} implementation cost not included. $2$ calls to \texttt{add} are required. \cref{fig:M-1-implementation-quantum-circuit}.\\
      \hline
      $V_{-1}$ & $2 \left( n-1 \right)$ & $2 + \left[ 0, n-1 \right]$ & $4n + \left[ 0, n-1 \right]$ \texttt{$X$} & $n-1$ borrowed & \cref{fig:V-1-implementation-quantum-circuit}.\\
      \hline
      $S_{-1}$ & $16n + 1$ & $0$ & $5 + 8 \left[ 0, n \right]$ \texttt{$X$} & $n$ borrowed & \cref{fig:S-1-implementation-quantum-circuit}.\\
      \hline
    \end{tabular}%
  }
  \caption{Precise gate count for the most important subroutines used in the quantum implementation of the wave equation solver. $n$ always reprensent the size of the input(s), except for the $n$-controlled \texttt{CNOT} where $n$ is the number of controls. When the number of gates depends on a generation-time value, the range of all the integer values possible is shown with square brackets. For example, $\left[ 0, n-1 \right]$ means that, depending on the generation-time value provided, the number of gates will be an integer between $0$ and $n-1$ included. \ket{0}-init ancillas represent the standard ancilla-type: qubits that are given in the state \ket{0} and should be returned in that exact same state. On the other side, borrowed ancillas can be given in any state and should be returned in the exact same state they were borrowed in.}\label{tab:gate_counts}
\end{table}

\begin{table}[h]
  \centering
  \resizebox{\linewidth}{!}{%
    \begin{tabular}{|l|c|c|c|c|m{2.1cm}|}
      \hline
      Unitary & Toffoli count & \texttt{CNOT} count & $1$-qubit gate count & \# ancillas & notes\\
      \hline
      \hline
      \texttt{$e^{-i H_1 t}$} & $22n - 12$ & $28n + 7 + 3 \left[ 0, n-1 \right]$ &
                                                                                   \begin{tabular}{@{}c@{}}
                                                                                     \begin{tabular}{@{}c c@{}}
                                                                                       $3n$ \texttt{$H$} & $3n$ \texttt{$S$} \\
                                                                                       $2n$ \texttt{$T$} & $36n$ \texttt{$P_h$} \\
                                                                                     \end{tabular} \\
                                                                                     $30n + 10 + 2\left[ 0, n-1 \right]$ \texttt{$X$}
                                                                                   \end{tabular} &
                                                                                                   \begin{tabular}{@{}c@{}}
                                                                                                     $1$ \ket{0}-init \\
                                                                                                     $n-1$ borrowed
                                                                                                   \end{tabular}
              & \texttt{add} implementation cost not included. $4$ calls to \texttt{add} are required.\\
      \hline
      \texttt{$e^{- i H_{-1} t}$} & $38n - 11$ & $28n + 7 + 3 \left[ 0, n-1 \right]$ &
                                                                                       \begin{tabular}{@{}c@{}}
                                                                                         \begin{tabular}{@{}c c@{}}
                                                                                           $3n$ \texttt{$H$} & $3n$ \texttt{$S$} \\
                                                                                           $2n$ \texttt{$T$} & $36n$ \texttt{$P_h$} \\
                                                                                         \end{tabular} \\
                                                                                         $30n + 15 + 10\left[ 0, n \right]$ \texttt{$X$}
                                                                                       \end{tabular} & \begin{tabular}{@{}c@{}}
                                                                                                         $1$ \ket{0}-init \\
                                                                                                         $n-1$ borrowed
                                                                                                       \end{tabular} & \texttt{add} implementation cost not included. $4$ calls to \texttt{add} are required.\\

      \hline
      \texttt{$e^{- i H t}$} & $82n - 35$ & $84n + 21 + 9 \left[ 0, n-1 \right]$ &
                                                                                   \begin{tabular}{@{}c@{}}
                                                                                     \begin{tabular}{@{}c c@{}}
                                                                                       $9n$ \texttt{$H$} & $9n$ \texttt{$S$} \\
                                                                                       $6n$ \texttt{$T$} & $108n$ \texttt{$P_h$} \\
                                                                                     \end{tabular} \\
                                                                                     $90n + 35 + 14\left[ 0, n \right]$ \texttt{$X$}
                                                                                   \end{tabular} & \begin{tabular}{@{}c@{}}
                                                                                                     $1$ \ket{0}-init \\
                                                                                                     $n-1$ borrowed
                                                                                                   \end{tabular} & \texttt{add} implementation cost not included. $12$ calls to \texttt{add} are required.\\
      \hline
      
    \end{tabular}
  }
  \caption{Number of gates and ancillas needed to simulate the easy-to-simulate Hamiltonians $H_1$ and $H_{-1}$ that are part of the decomposition of $H$ as well as $e^{-iHt}$. It is important to realise that the gate counts for $e^{-iHt}$ are only valid up to a given $t$ or $\epsilon$ (once one is fixed, the value of the other can be computed). In order to make the gate count generic for any $t$ and $\epsilon$, the number of repetitions should be computed (see $n$ in \cref{eq:lie-trotter-suzuki-product-formula-timestep}). Note that some of the $\left[ 0, n-1 \right]$ ranges have been simplified to $\left[ 0, n \right]$ for conciseness.}\label{tab:HS_gate_counts}
\end{table}

\begin{table}[h]
  \centering
  \resizebox{\linewidth}{!}{%
    \begin{tabular}{|l|c|c|c|c|}
      \hline
      Adder used & Toffoli count & \texttt{CNOT} count & $1$-qubit gate count & \# ancillas\\
      \hline
      \hline
      \texttt{add\_qft} & $82n - 35$ & $144n^2 - 60n$ &
                                                        \begin{tabular}{@{}c@{}}
                                                          \begin{tabular}{@{}c c@{}}
                                                            $24n^2 + 25n$ \texttt{$H$} & $9n$ \texttt{$S$} \\
                                                            $48n^2 - 42n$ \texttt{$T$} & $108n$ \texttt{$P_h$} \\
                                                          \end{tabular} \\
                                                          $18n^2 - 18n$ \texttt{$R_n$} \\
                                                          $114n + 35 + 14\left[ 0, n \right]$ \texttt{$X$}
                                                        \end{tabular} &
                                                                        \begin{tabular}{@{}c@{}}
                                                                          $2$ \ket{0}-init \\
                                                                          $n-1$ borrowed
                                                                        \end{tabular}\\
      \hline
      \texttt{add\_arith} & $222n - 175$ & $348n + 21 + 9 \left[ 0, n-1 \right]$ &
                                                                                   \begin{tabular}{@{}c@{}}
                                                                                     \begin{tabular}{@{}c c@{}}
                                                                                       $9n$ \texttt{$H$} & $9n$ \texttt{$S$} \\
                                                                                       $6n$ \texttt{$T$} & $108n$ \texttt{$P_h$} \\
                                                                                     \end{tabular} \\
                                                                                     $90n + 35 + 14\left[ 0, n \right]$ \texttt{$X$}
                                                                                   \end{tabular} & \begin{tabular}{@{}c@{}}
                                                                                                     $n$ \ket{0}-init \\
                                                                                                     $n-1$ borrowed
                                                                                                   \end{tabular}\\
      \hline
    \end{tabular}
  }
  \caption{Number of gates and ancillas needed to simulate the Hamiltonian used to solve the $1$-dimensional wave equation depending on the adder implementation used. It is important to realise that the gate counts for $e^{-iHt}$ reported in this table are only valid up to a given $t$ or $\epsilon$ (once one is fixed, the value of the other can be computed). In order to make the gate count generic for any $t$ and $\epsilon$, a number of repetitions $r$ should be computed (named $n$ in \cref{eq:lie-trotter-suzuki-product-formula-timestep} and studied in~\cite[arXiv: Appendix F]{1711.10980v1} and~\cite{1912.08854v1}). Note that the gate counts have been simplified by removing negligible terms when possible.}\label{tab:final_HS_gate_counts}
\end{table}

\subsection{Impact of the precision requirements}
\label{sec:impact-prec-requ}

The gate counts presented in~\cref{tab:gate_counts},~\cref{tab:HS_gate_counts} and~\cref{tab:final_HS_gate_counts} are only valid when the precision of the solver is not accounted for. When the solver precision matters, an additional step that consists is splitting the Hamiltonian Simulation into $r$ steps needs to be performed as noted in~\cite[arXiv: Appendix F]{1711.10980v1}.

Several bounds exist to determine a $r\in\mathbb{N}^*$ that will analytically ensure that the
maximum allowable error $\epsilon$ is not exceeded. The definition of such bounds can be found
in~\cite[arXiv: Appendix F]{1711.10980v1} and~\cite{1912.08854v1}.

The first bound has been devised by analytically bounding the error of simulation due to the
Trotter-Suzuki formula approximation by $\epsilon_0$
\begin{equation}
  \label{eq:error-to-bound}
  \left\vert \left\vert \exp\left( -it \sum_{j=0}^{m-1} H_j \right) - \left[ S_{2k}\left( -\frac{it}{r} \right) \right]^r\right\vert \right\vert \leqslant \epsilon_0
\end{equation}
and then let $\epsilon_0 \leqslant \epsilon$ for a given desired precision $\epsilon$.
If we let $\Lambda = \max_j \vert\vert H_j\vert\vert$ and
\begin{equation}
  \label{eq:tau-factor-definition}
  \tau = 2m5^{k-1}\Lambda\vert t \vert  
\end{equation}
then 
\begin{equation}
  \label{eq:analytic_bound}
  r^{ana}_{2k} = \left\lceil \max\left\{ \tau, \sqrt[2k]{\frac{e\tau^{2k+1}}{3\epsilon}} \right\} \right\rceil.
\end{equation}
This bound is called the \emph{analytic bound}.

A better bound called the \emph{minimised bound} can be devised by searching
for the smallest possible $r$ that satisfies the conditions detailed in
\cite[Propositions F.3 and F.4]{1711.10980v1}. This bound is rewritten in Equation
\eqref{eq:minimised-bound-definition}.

\begin{equation}
  \label{eq:minimised-bound-definition}
  r^{min}_{2k} = \min \left\{ r \in \mathbb{N}^* : \frac{\tau^{2k+1}}{3r^{2k}} \exp\left( \frac{\tau}{r} \right) < \epsilon \right\}
\end{equation}

Another bound involving nested commutators of the $H_i$ is described in~\cite{1912.08854v1} and gives
\begin{equation}
  \label{eq:r_nested_commutators}
  r_{2k}^{comm} \in \bigO{\frac{\alpha_{comm}^{\frac{1}{2k}} t^{1 + \frac{1}{2k}}}{\epsilon^{\frac{1}{2k}}}}
\end{equation}
where $k$ is the order of the product-formula used, $t$ the time of simulation, $\epsilon$ the error and
\begin{equation}
  \label{eq:alpha_commutator}
  \alpha_{comm} = \sum_{i_0, i_1, \dots{}, i_p = 0}^{m-1} \left\vert \left\vert \left[ H_{i_p}, \dots{} \left[ H_{i_1}, H_{i_0} \right] \right] \right\vert \right\vert.
\end{equation}

Once the value of $r$ has been computed, the quantum circuit simulating the matrix $H$ for a time $\frac{t}{r}$ should be repeated $r$ times. This adds a factor of $r$ in front of all the gate counts computed in~\cref{tab:gate_counts}, \cref{tab:HS_gate_counts} and~\cref{tab:final_HS_gate_counts}.

\subsection{Impact of error-correction}
\label{sec:impact-error-corr}

When error-correction is studied, two gates are particularly important: $T$ and Toffoli gates. The $T$ gate has a prohibitive cost when compared to the Clifford quantum gates and implementing a Toffoli gate requires $7$ of such $T$ gates as noted in~\cite{Fowler2012} and~\cite[Fig. 1]{shende2008cnotcost}.

\cref{tab:T_gate_estimations} summarise the cost of the non Clifford quantum gates used in the implementation of the $1$-dimensional wave equation solver. The rotation gates need to be approximated. One solution to approximate the $R_n$ and $P_h$ gates is given in~\cite{ross2014optimal}. In order to obtain practical results as opposed to theoretical ones, we chose to use the number computed in~\cite[Table 1]{Kim2018}.

The final $T$-count is summarised in~\cref{fig:T_count}. From~\cref{fig:T_count_plot} it is clear that the \texttt{add\_arith} implementation is more efficient than the \texttt{add\_qft} one. 

\begin{table}[h]
  \centering
  \begin{tabular}{|l|c|m{5cm}|}
    \hline
    Gate & $T$ count & Notes\\
    \hline
    $T$ & 1 & \\
    \hline
    $S$ & 2 & \\
    \hline
    \texttt{CCNOT} & $7$ & See~\cite{Fowler2012}.\\
    \hline
    $P_h$ & 379 & $\epsilon = 10^{-15}$, approximated from~\cite{Kim2018}. \\
    \hline
    $R_n$ & 379 & $\epsilon = 10^{-15}$, approximated from~\cite{Kim2018}. \\
    \hline
  \end{tabular}
  \caption{$T$-gate cost of the non Clifford quantum gates used in the wave equation solver implementation.}\label{tab:T_gate_estimations}
\end{table}

\begin{figure}[h]
  \centering
  \subfigure[%
  {Number of $T$-gates needed to simulate the Hamiltonian used to solve the $1$-dimensional wave equation depending on the adder implementation used. Based on~\cref{tab:final_HS_gate_counts} and~\cref{tab:T_gate_estimations}.}%
  ]{%
    \includegraphics[width=.35\linewidth]{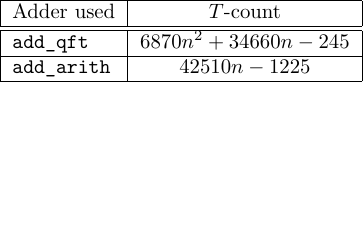}%
    \label{tab:T_count_table}
  }\hspace{.03\linewidth}
  \subfigure[%
  {Plot of the $T$-count devised in \cref{tab:T_count_table} for the two different adder implementations.}%
  ]{%
    \includegraphics[width=.6\linewidth]{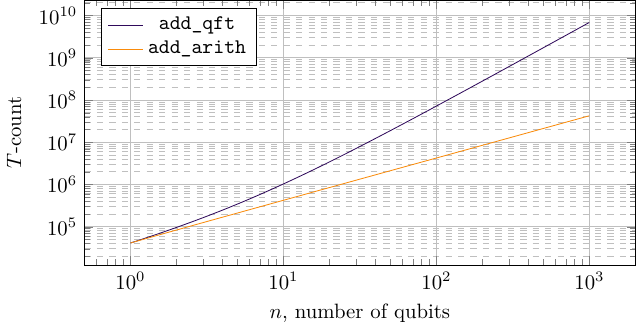}%
    \label{fig:T_count_plot}
  }  
  \caption{Analysis of the $T$-count of the $1$-dimensional wave equation solver quantum implementation with respect to the adder implementation used.}\label{fig:T_count}
\end{figure}

\end{document}